\documentclass[11pt]{llncs}

\usepackage[a4paper, margin=2.5cm]{geometry}

\usepackage[T1]{fontenc}
\usepackage{algorithm}
\usepackage[spaceRequire=false,beginComment=//~,beginLComment=/*~,endLComment=~*/,
noEnd=true,rightComments=false,indLines=false]{algpseudocodex}
\usepackage{multicol}
\usepackage{mathtools}
\usepackage{subcaption}
\usepackage{amssymb}
\usepackage[title]{appendix}

\newcommand{\recoverycode}[1]{\textcolor{blue}{#1}}

\newcommand{\pwb}[1]{\textbf{pwb(#1)}}

\newcommand{\psync}[1]{\textbf{psync(#1)}}
\newcommand{\pwbi}{\mbox{\texttt{pwb}}}
\newcommand{\pfencei}{\mbox{\texttt{pfence}}}
\newcommand{\psynci}{\mbox{\texttt{psync}}}

\newcommand{\CAS}{\mbox{\textit{CAS}}}

\newcommand{\FAI}{\mbox{\textit{FAI}}}
\newcommand{\TAS}{\mbox{\textit{TAS}}}
\newcommand{\True}{\mbox{\texttt{true}}}
\newcommand{\False}{\mbox{\texttt{false}}}
\newcommand{\Head}{\texttt{Head}}
\newcommand{\First}{\texttt{First}}

\newcommand{\Tail}{\texttt{Tail}}
\newcommand{\Last}{\texttt{Last}}
\newcommand{\Q}{\mathcal{Q}}
\newcommand{\LinQ}{\texttt{LinQ}}
\newcommand{\nOps}{\ensuremath{\texttt{nOps}}}

\newcommand{\nd}{\ensuremath{\texttt{nd}}}
\newcommand{\f}{\ensuremath{\texttt{f}}}

\newcommand{\LCRQ}{{\scshape LCRQ}}
\newcommand{\PLCRQ}{{\scshape PerLCRQ}}
\newcommand{\CRQ}{{\scshape CRQ}}
\newcommand{\PCRQ}{{\scshape PerCRQ}}
\newcommand{\IQ}{{\scshape IQ}}

\newcommand{\PIQ}{{\scshape PerIQ}}
\newcommand{\FIFO}{{\scshape FIFO}}
\newcommand{\NULL}{\mbox{\scshape Null}}

\newcommand{\PBQueue}{\mbox{\scshape PBQueue}}
\newcommand{\PWFQueue}{\mbox{\scshape PWFQueue}}

\newcommand{\Recovery}{\Call{Recovery}{}}

\newcommand{\localHead}{\mathit{h}}
\newcommand{\localTail}{\mathit{t}}
\newcommand{\itemX}{\mathit{x}}
\newcommand{\crqV}{\texttt{crq}}
\newcommand{\crqVar}{\texttt{}}
\newcommand{\EMPTY}{\texttt{EMPTY}}
\newcommand{\CLOSED}{\texttt{CLOSED}}
\newcommand{\OK}{\texttt{OK}}
\newcommand{\cntb}{count_{\bot}}

\newcommand{\var}[1]{\mathit{#1}}

\newcommand{\Enqueue}{\mbox{\sc Enqueue}}
\newcommand{\Dequeue}{\mbox{\sc Dequeue}}

\newcommand{\remove}[1]{}

\newcommand{\Youla}[1]{}
\newcommand{\Y}[1]{}
\newcommand{\Note}[1]{}

\newcommand{\here}[1]{{\color{teal} #1}\normalcolor}
\newcommand{\y}[1]{{\color{blue} #1}\normalcolor}
\newcommand{\nick}[1]{}

\newcommand{\PWFqueue}{{\sc PWFqueue}}

\newcommand{\PBqueue}{{\sc PBqueue}}

\newcommand{\OneFile}{{\sc OneFile}}
\newcommand{\CXPUC}{{\sc CX-PUC}}
\newcommand{\CXPTM}{{\sc CX-PTM}}
\newcommand{\RedoOpt}{{\sc RedoOpt}}
\newcommand{\FHMP}{{\sc FHMP}}

\newcommand{\CAPSULES}{{\sc Capsules-Normal}}
\newcommand{\Romulus}{{\sc Romulus}}

\newtheorem{scenario}{Scenario}

\begin{document}
\title{Highly-Efficient Persistent FIFO Queues\thanks{Supported by the Hellenic Foundation for Research and Innovation (HFRI) under the ``Second Call for HFRI Research Projects to support Faculty Members and Researchers'' (project number: 3684).}}
%
%
\author{Panagiota Fatourou\inst{1,2} \and
Nikos Giachoudis\inst{1} \and
George Mallis\inst{1,2}}
\authorrunning{P. Fatourou et al.}
%
\institute{FORTH ICS, Greece
\email{\{faturu,ngiachou\}@ics.forth.gr} \and
University of Crete, Greece
\email{csd4165@csd.uoc.gr}}
\maketitle              
\begin{abstract}
In this paper, we study the question whether techniques employed, in a conventional system, 
by state-of-the-art concurrent algorithms to avoid contended hot spots are still efficient 
for recoverable computing in settings with Non-Volatile Memory (NVM). We focus on concurrent FIFO queues that have two end-points, 
head and tail, which are highly contended. 

We present a persistent FIFO queue implementation that performs a pair of persistence instructions 
per operation (enqueue or dequeue). The algorithm achieves to perform 
these instructions on variables of low contention by employing Fetch\&Increment and using 
the state-of-the-art queue implementation by Afek and Morrison (PPoPP'13). These result in performance 
that is up to 2x faster than state-of-the-art persistent FIFO queue implementations. 

\end{abstract}

\section{Introduction}

Non-Volatile Memory (NVM) has been proposed 
as an emerging memory technology to provide the persistence 
capabilities of secondary storage at access speeds that are 
close to those of DRAM. 
In systems with NVM, concurrent algorithms can be designed to be {\em recoverable} (or {\em persistent}),
by persisting parts of the data they use. This enables
to restore their states after system-crash failures. 
In this direction, a vast amount of papers that appear in the literature
propose persistent implementations for a wide collection of concurrent data structures,
including stacks~\cite{FKK22,RA21}, queues~\cite{FKK22,FH+18,SP21}, heaps~\cite{FKK22}, and many others.
Moreover, a lot of papers have focused on designing general transformations~\cite{AB+22,NormOptQueue19,FKK22}
and universal constructions~\cite{CFR18,CFP20eurosys,FKK22,RC+19} that can be applied on many conventional concurrent data
structures 
to get their persistent analogs. (This list of references is by no means exhaustive.)

In this paper, we study the question whether techniques employed, in a conventional system,
by state-of-the-art concurrent algorithms, to avoid contended hot spots, 
are still efficient for persistent computing in settings with NVMs.
We focus on concurrent queues that 
have two end-points, head and tail, which are highly contended.
Accesses to the queue's endpoints result in bottlenecks, making most existing concurrent queue implementations non-scalable. 
Several papers have proposed techniques for reducing this cost. 
One approach uses software combining~\cite{FK12ppopp,FK11spaa,FK14,FK17,KSW18,HIST10}, a technique that attempts to 
reduce synchronization by having a thread, the {\em combiner}, to collect and apply
operations (of the same type) on each of the endpoints. Other papers~\cite{FK11spaa,FK14,FKR18,LCRQ13} 
use \emph{Fetch\&Increment} (\FAI) to avoid hotspots; \FAI\ is an \emph{atomic}
primitive which increases by one the value of a shared variable it takes as a parameter, and returns its previous
value.
Avoiding hotspots using \FAI, when designing a concurrent FIFO queue, resulted in \LCRQ,
the state-of-the-art concurrent queue implementation, for which it has been experimentally shown~\cite{LCRQ13} 
to significantly outperform the previously-proposed state-of-the-art algorithm~\cite{FK12ppopp}, 
which was based on software-combining.

The performance power of software combining in persistence has been studied in~\cite{FK12ppopp}. 
The combining-based persistent \FIFO\ queue implementation~\cite{FK12ppopp} has been shown to 
outperform all other persistent queue implementations,
specialized or not. As specialized algorithms take into consideration the semantics 
of the data structure to make design decisions in favor of performance, this was
kind of surprising. 
However, the specialized persistent queues that are proposed in the literature~\cite{FH+18,SP21} are not based 
on state-of-the-art queue implementations~\cite{LCRQ13,FKK22}.

In this paper, we examine whether the use of \FAI\ is a more promising approach
for designing persistent queues than employing combining.
Specifically, we focus on the state-of-the-art queue implementation (\LCRQ)~\cite{LCRQ13},
and study the correctness and performance implications of different persistence
choices to make it persistent. We end up with a persistent
FIFO queue implementation, called \PLCRQ, that performs much better
than any previous such implementation.

\LCRQ\ is inspired by the simple idea of using an infinite array, $\Q$ (which initially contains $\bot$), 
and two \FAI\ objects, \Tail\ and \Head\ (initially $0$).
An enqueuer executes a \FAI\ on \Tail\ to get an index $i$ of the array where the new item should be 
enqueued. 
Then, it performs a \Call{Get\&Set}{} to swap the element to be inserted with the current value of 
$\Q[i]$ (which in a successful enqueue it should be $\bot$). 
Similarly, a dequeuer executes \FAI\ on \Head, to get the array index  from where it will dequeue
an item by swapping the value of that element with $\top$. This way, each element of the array
has at most one enqueuer that tries to insert an element in it and at most one dequeuer that tries to 
dequeue from it. 
This simple algorithm, which we call \IQ, is not practical, as it requires an array of 
infinite size. Moreover, it may result in a livelock with an enqueuer and a dequeuer
always interfering with one another without ever making progress. 
\CRQ~\cite{LCRQ13} follows similar ideas as \IQ\ but it uses a circular array (of bounded size)
to store the new items. This introduces many more cases where synchronization 
is needed between enqueuers and dequeuers, but results in a more practical algorithm. To solve the live-lock 
problem (and the space limitation of the static array),
\LCRQ\ employs several \CRQ\ instances, connected in a linked list to form a queue
as in the well-known lock-free queue implementation in~\cite{MSQUEUE96}. 
(Section~\ref{sec:iq2lcrq} presents the details of \IQ, \CRQ\ and \LCRQ.)

Our effort focuses on making \LCRQ\ persistent with the least possible {\em persistence cost} (Section~\ref{sec:main}),
i.e. without paying a high performance overhead to support persistence, in periods of time where no failures occur.  
A shared variable can be persisted by executing a pair of persistence instructions: \pwbi, which requests 
the system to flush back the current value of a shared variable to NVM and works in a non-blocking way, 
and \psynci, that blocks until flushing has been completed for all preceeding \pwbi s. 
Persistence can easily be supported on top of an existing concurrent algorithm by persisting all shared variables~\cite{DBLP:conf/wdag/IzraelevitzMS16}
every time they are updated. 
However, this technique would have excessive performance overhead. 
On the other hand, the simple approach of persisting only \Head\ (\Tail) every time a \FAI\ is executed on it, 
results also in excessive performance cost as it violates two major persistence principles
crucial for performance~\cite{AB+22}: a) it executes many persistence instructions per operation, 
and b) it applies persistence instructions on highly contended shared variables accessed by all threads.

We focused first on achieving persistence with a single pair of persistence instructions per operation, which is optimal. 
Starting from \IQ, we managed to persist only once per operation.  
This required careful design of the algorithms' {\em recovery functions} (the recovery function of an implementation
is executed by the system, after a crash, to bring the data structure in a coherent state). 
It also required  long argumentation to prove that the proposed
algorithms are indeed correct (i.e., they satisfy {\em durable linearizability}~\cite{IMS16}). 
We consider these to be main contributions of the paper, as the persistence part of the algorithms
themselves (without considering the recovery functions) had to be maintained as simple and minimal as possible for achieving our performance goals. 
We next observed that even when we persist only once, we cannot beat the state-of-the-art combining based approaches~\cite{FKK22},
if persistence instructions are executed on \Head\ and \Tail\ that are highly contended. So,  we came up with persistent variants of the algorithms, 
where an enqueue operation $op$ persists only the last value it writes into $Q$. Moreover, for 
\IQ, we came up with \PIQ, where dequeue operations persist also only the last value they write into $Q$. As each position of the array 
is accessed by only two threads, both persistence principles mentioned above are respected by \PIQ.
This results in low performance overhead. 

Respecting the persistence principles~\cite{AB+22}, when designing \PCRQ\ (the persistent version of \CRQ) was more challenging. 
A position $i$ of $Q$ can now be accessed by any operation that reads
$i + kR$ in \Tail\ (or \Head), where $k \geq 0$ is an integer and $R$ is the size of the circular array. 
This leads to the necessity to persist \Head\
every time an item is dequeued to ensure durable linearizability. However, as explained previously,
doing so has considerable performance overhead. To avoid this cost, we introduce the
following technique. Each thread maintains a local copy of \Head. Every time it 
is about to complete a dequeue operation, it persists the local copy of \Head\
(instead of its shared version, on which \FAI\ is executed to get an index in $Q$).
Since the local copies of \Head\ are single-writer, single-reader variables,
this significantly reduces the cost of persistence (Section~\ref{sec:evaluation}, Figures~\ref{fig:plcrq-comp},~\ref{fig:over-example}).
We believe that this technique is of independent interest and can be used to get
persistent versions of many other state-of-the-art concurrent algorithms. 

Our experimental analysis (Section~\ref{sec:evaluation}) shows that \PLCRQ\ is at least 2x faster than its best competitor, \PBqueue~\cite{FKK22}.
We also provide experiments to support all our claims above.  

Performing a single pair of persisting instructions per operation in \PIQ\ introduces some recovery cost.
(The {\em recovery cost} is the time needed to execute the recovery function.)
The {\em recovery function} of \PIQ\
has to scan $Q$ until it finds a streak of $n$ subsequent $\bot$ values. This might end up to be expensive,
as bigger parts of $Q$ has to be scanned as the time is progressing. On the other hand,
we may choose to trade performance at normal execution time (where no failures occur)
by periodically persisting \Head\ and \Tail. This illustrates an interesting tradeoff between 
the performance at normal execution time and the recovery cost: better performance at normal execution time results in higher cost at recovery
(and vice versa).  It is an interesting problem to study whether
this trade-off appears when designing the persistence versions of other state-of-the-art techniques for
ensuring synchronization and for designing concurrent data structures in a conventional (non-NVM) setting.
We believe that this trade-off is inherent in many cases.

We provide experiments  (Section~\ref{sec:evaluation}) 
to simulate failures which allows us to measure the recovery cost. 
To the best of our knowledge, none of the papers on persistent data structures in shared memory systems
provide experiments to measure the recovery cost. 

Summarizing, the main contributions of this paper are the following: \\
{\bf (1)} We present \PLCRQ, a persistent implementation of a \FIFO\ queue that 
is much faster than existing state-of-the-art persistent queue implementations. \PLCRQ\
promotes the idea of starting from the state-of-the-art implementation of a concurrent data structure (DS) to come up with an efficient persistent implementation of it. 
This departs from the approach followed in previous specialized persistent FIFO queue implementations~\cite{FH+18,SP21}.\\
{\bf (2)} We illustrate an interesting tradeoff between the persistence cost of an algorithm at normal execution time and its recovery cost.  \\
{\bf (3)} We provide a framework to simulate failures and measure the recovery cost. \\
{\bf (4)} Our experimental analysis shows that a) \PLCRQ\ is at least two times faster than its competitors, 
and that b) respecting the persistence principles presented in~\cite{AB+22} is indeed crucial for performance. \\
{\bf (5)}  Based on our experiments, we propose a new technique for reducing the persistence cost,
namely using local copies of highly-contended shared variables by each thread and persisting them instead.  
This technique
transfers part of the persistence cost from normal execution time to recovery, and is of independent interest. 


\section{Preliminaries}
We consider a system where $n$ asynchronous threads run concurrently and communicate using shared objects.
In addition to read/write objects, that support atomic {\em reads} and {\em writes} on/to the state of the object,
the system also provides the following {\em atomic} primitives: 
a) \Call{Fetch\&Increment}{$V$} (\FAI{}), which adds $1$ to $V$'s current value  and returns the value that $V$ 
had before the increment, b) \Call{Get\&Set}{$V$, $v'$}, which stores value $v'$ into $V$ 
and returns $V$'s previous value, c) \Call{Compare\&Swap}{$V, v_{old}, v_{new}$} (\CAS{}), 
which checks the value of $V$ and if it
is $v_{old}$, it changes its value to $v_{new}$ and returns \True{}, otherwise it makes no change
and returns \False{}. A variant of \CAS\ is 
\Call{CAS2}{$V, ( v_0^{old}, v_1^{old}), ( v_0^{new}, v_1^{new})$}, which operates atomically 
on an array $V$ of two elements. Specifically, it checks if
$V[0] == v_0^{old} \land V[1] == v_1^{old}$ and if this is true it changes the value of $V$
to $ \{ v_0^{new}, v_1^{new} \}$ and returns \True{}; otherwise it returns \False{}
without changing $V$.
Finally, \Call{Test\&Set}{$B$} changes the value of a shared bit $B$ to $1$ and returns its
value before the change; it comes together with  \Call{Reset}{$B$}, which resets the value of $B$ to $0$.
In the pseudocodes in later sections, the names of shared objects start with a capital letter. 
We assume the \emph{Total Store Order} (\emph{TSO}) model,
where writes become visible in program order.

The system's memory is comprised of a non-volatile part and its volatile components
(registers, caches, DRAM). 
We study persistence under the {\em full-system crash failure} (or {\em system-wide crash failure}) model
\cite{izraelevitz2016linearizability}.
When a failure occurs, the values of all variables 
stored in volatile memory are lost and are reset to their initial values at recovery time.
On the contrary, any data that was written back ({\em persisted}) to NVM persist (i.e., they are available
at recovery time).

We assume {\em explicit epoch persistency} \cite{izraelevitz2016linearizability}, where a 
write-back to the persistent memory is initiated through a specific instruction called {\em persistent 
write-back} (\pwbi{}); \pwbi{} is asynchronous, i.e., the order in which \pwbi{}s take effect 
may not be preserved inherently. To enforce ordering when necessary, a \pfencei{} instruction  
ensures that all \pwbi{} instructions performed before the \pfencei{}, are realized before those 
that follow it. However, \pfencei{} is also asynchronous, so it does not provide any guarantee that 
the \pwbi{}s that preceded it have occured by the time of a crash.  
A \psynci{} instruction causes a thread to block until all prior \pwbi{} 
instructions have been realized (i.e., the data they flush have indeed been written back to the NVM). 
We collectively refer to \pwbi{}, \pfencei{}, and \psynci{} as the set of {\em persistence 
instructions}. The {\em persistence cost} of an operation's instance (an execution) 
is the time overhead incurred from executing its persistence instructions. 

We assume that a recoverable implementation has 
an associated {\em recovery function} (for the system as a whole). After a system failure, the recovery 
function is invoked to bring the system in a consistent state. Then, the system may (asynchronously) recover 
failed threads, which start by invoking new operations. 
The {\em recovery cost} of an implementation is the 
time needed to execute its recovery function. We usually measure it by taking the average among 
several runs, in each of which the recovery function is executed at the end 
after a system-wide failure (see details in Section~\ref{sec:evaluation}).

Consider an execution $\alpha$, and let $c_1, c_2, \ldots$ be crash events in $\alpha$, in order. 
We divide $\alpha$ into {\em epochs} $E_1, E_2, \ldots$.
The first epoch starts at $C_0$ and ends with $c_1$, if $C_1$ exists,
otherwise $E_1$ is the entire execution. 
For each $k > 0$, epoch $E_k$ starts with the event following $c_{k-1}$ and ends with $c_k$, if 
$c_k$ exists in $E_k$. 
If $c_k$ does not exist, $E_k$ is the suffix of $\alpha$ after $c_{k-1}$.

An operation $op$ has an {\em invocation} and may have a {\em response}; $op$'s  
{\em execution interval} starts with its invocation and ends with its response (if it exists)
or it is the suffix of an execution starting with its invocation (otherwise). 
An implementation is {\em lock-free}, if in every infinite execution it produces, 
it holds that if the number of failures is finite, an infinite number of operations respond.

In a conventional setting where threads never recover after a crash, an execution $\alpha$ is 
{\em linearizable}~\cite{10.1145/78969.78972}
if there exists a sequential execution $\sigma$ which contains all completed
operations in $\alpha$ (and possibly some of the uncompleted ones), such that 
a) the response of every operation in $\alpha$ is the same as that of the corresponding
operation in $\sigma$, and b) the order of operations in $\sigma$ respects the
partial order induced by the execution intervals of operations in $\alpha$.  
An algorithm is {\em linearizable}, if all the executions it produces are linearizable. 

We now consider a setting where failed threads may recover. Briefly, an execution is {\em durably 
linearizable}~\cite{izraelevitz2016linearizability}, 
if the state of the object after every crash reflects all the changes performed by the operations
that have completed by the time of the crash (and possibly by some of the uncompleted operations). 
Our goal is to design durably linearizable implementations of FIFO queues with low persistence cost.

\section{LCRQ: A Brief Description} \label{sec:iq2lcrq}
\LCRQ~\cite{LCRQ13} uses multiple instances of a static circular FIFO queue implementation, called \CRQ.
\CRQ\ is based on a simple but impractical algorithm, which we call \IQ\ and we discuss it first.

\noindent
{\bf The \IQ\ algorithm.}
In \IQ\ (black lines in Algorithm~\ref{alg:piq}), 
the queue is represented as an array ($\Q$) of infinite size, initialized with $\bot$ in all of its cells, 
and two variables \Head{} and \Tail{}, initially 0.
\Head{} stores the index of $Q$'s item that is going to be dequeued next, 
whereas \Tail{} stores the index of the next available $Q$'s cell to store a new item.

Enqueue uses \FAI{} to read \Tail{} 
and get the current available index of $Q$. It then attempts to store the new item there
using \Call{Get\&Set}{}. If the  \Call{Get\&Set}{} returns $\bot$, the enqueue is successful.
Otherwise, it re-tries by applying the steps described above. 
A dequeue uses \FAI{} to identify the next item to be removed by reading \Head. 
It then attempts to read that item
and replace it with $\top$ by using \Call{Get\&Set}{}. If \Call{Get\&Set}{}
returns a 
not $\bot$ item, the dequeue is successful and returns it. 
If \Call{Get\&Set}{} returns $\bot$ and $\Head > \Tail$,
the dequeue returns \EMPTY{} (i.e., the queue is empty). 
Otherwise, these steps are repeated until either some non-$\bot$ value or \EMPTY{} is returned.

The \IQ{} algorithm is very simple but unrealistic, because of the infinite table it uses. 
Another issue of the \IQ{} algorithm is that  
 if a dequeue operation continuously swaps $\top$ in cells that an enqueue 
operation is trying to access, 
we would have a {\em livelock}. 

\noindent
{\bf The \CRQ\ algorithm.}
To address the infinite array issue,  
\CRQ{} (see black lines of Algorithm~\ref{alg:pcrq}) implements $\Q{}$ as a circular array, of size $R$.  
As in \IQ{}, \Head{} and \Tail{} are implemented as \FAI{} objects, whose values 
are incremented indefinitely. Each enqueue (dequeue) reads the value $i$ in \Tail{} (\Head).
Consider an enqueue (dequeue) operation  that reads index $i$ by executing \FAI\ on \Tail\ (\Head) 
and stores (removes) an element
in (from) $\Q[i \bmod R]$.  We say that the {\em index} of the operation is $i$. A pair of enqueue and dequeue operations
that both have index $i$ are called {\em matching operations}, and we denote them by $\var{enq}_i$ and $\var{deq}_i$.

Every cell in the array now stores a $3$-tuple object $(s, i, v)$, containing the {\em safe bit}, an {\em index}, and the {\em value} 
of the item it stores; its initial value is $(1, u, \bot)$,  where 
$u$ is the actual index of each cell of the array. 
If the node's value is $\bot$ 
then that node is \emph{unoccupied}, otherwise it is \emph{occupied}. 
An enqueue with index $i$ and value $v$ tries to store the triplet $( 1, v, i )$ in
$\Q[i \bmod R]$, if this array cell is unoccupied.
\CRQ\ ensures that only the dequeue with index $i$ can dequeue this item. 

The use of a circular array results in several synchronization challenges between enqueuers and dequeuers.
Specifically, 
\CRQ\ copes with the following cases: a) a dequeue $\var{deq}$ 
finds an element of $\Q$ it accesses to store an index less than or equal to the index it uses 
and be unoccupied ({\em empty transition}), and b) $\var{deq}$ finds $\Q$'s element occupied by an item 
that an enqueue other than its matching one has inserted ({\em unsafe transition}).
A {\em transition} is a change of the value of an array cell. A dequeue operation 
$\var{deq}$ performs a {\em dequeue transition} if it successfully executes the \Call{CAS2}{} in line~\ref{pcrq:dqtran},
changing the value of item $\Q[t \bmod R]$ from 
$(s, t, v)$ to $(s, t + R, \bot)$. 
In order to have that transition, $t$ should match the index $i$ of the cell, 
and the cell must be also occupied (line~\ref{pcrq:dqemptychk}). 
$\var{deq}$ executes an {\em empty transition} if it arrives before the enqueue that gets index $t$
stores its item in the cell, and finds the cell unoccupied.
Then, the \Call{CAS2}{} is called (line~\ref{pcrq:empty-1}) to change the index of that cell to
$t + R$.
This 
prevents an $\var{enq}$ that reads $i$ in \Tail\
to add its item in this array element. This is needed to 
ensure linearizability, 
as there will be no dequeuer to dequeue this item (a dequeuer has already {\em exhausted} index $i$).
Finally, $\var{deq}$ performs an {\em unsafe transition}, if it  
arrives while the cell is occupied with index $t-kR$, where $k>0$. 
Then, $\var{deq}$ executes the \Call{CAS2}{} of line~\ref{pcrq:unsafe} to change the safe bit to $0$.
This way it informs
later enqueuers to be careful, as the items they want to store in this position may never 
be removed. An enqueuer has to check whether $\Head\ > i$ and if this is so, the dequeuer 
may have already passed from this position, so they try the insertion with the next available index. 

Before attempting its insertion, a thread executing an \Call{Enqueue}{\crqV, $x$}, that reads $t$ in \Tail\ 
(line~\ref{pcrq:eqfaa}), checks that: a) the index 
$i$ of cell $\Q[t \bmod R]$ is at most as big as the value $t$ (line~\ref{pcrq:eqcas})
(otherwise, the enqueue operation tries the next potentially available cell, as another operation
with index higher than $t$ has already processed this item),
b) the safe bit $s$ is equal to $1$ or \Head\ is at most $t$
(line~\ref{pcrq:eqcas})  (the second condition ensures that the dequeue with index $t$ 
has not started its execution yet, so the signal through the safe bit does not concern it). 
An {\em enqueue transition} succeeds when the 
\Call{CAS2}{} operation in line~\ref{pcrq:eqcas} succeeds. 

Linearizability of \IQ\ and \CRQ\ are discussed in~\cite{LCRQ13}.
\CRQ\ implements a {\em tantrum} queue, i.e. a FIFO queue with relaxed semantics: 
any enqueue operation may return a special value \CLOSED\ (which identifies that 
the array is full or that it faces a livelock when trying to insert a new element in it); 
as soon as an enqueue returns \CLOSED\ on an instance of \CRQ, 
all other enqueue operations on the same \CRQ\ instance should also return \CLOSED.

\noindent
{\bf The LCRQ implementation.}
The main idea of \LCRQ{} (see black lines of Algorithm~\ref{alg:plcrq}) is to use a linked list of nodes, each representing an instance of \CRQ.
The linked list together with two pointers (\First\ and \Last) to its endpoints, 
implement a FIFO queue, which closely follows 
the lock-free \FIFO\ queue implementation by Michael and Scott~\cite{MSQUEUE96}.  
When an enqueue on the current instance of \CRQ\ returns \CLOSED,
it creates a new \CRQ\ instance containing the item to be inserted and appends
it to the list. 
When a dequeue on the \CRQ\ instance pointed by \First, returns \EMPTY, it dequeues this node 
by moving \First\ to point to the next \CRQ\ node in the list.
\LCRQ\ is a linearizable implementation of a FIFO queue~\cite{LCRQ13}. 


\section{Persistent FIFO Queue Algorithm} \label{sec:main}

\subsection{Persistent IQ} \label{sec:piq}

\begin{algorithm}[t]
\scriptsize
\caption{Persistent \IQ\ (\PIQ). Ignore code in blue to get the pseudocode for \IQ.}\label{alg:piq}
\begin{multicols}{2}
\begin{algorithmic}[1]
\Function{Enqueue}{$\itemX{}$: Item}
    \While{\True{}} \label{piq:loop}
        \State $\localTail{}\gets \Call{Fetch\&Increment}{\Tail{}}$\label{piq:faa}
        \If{$\Call{Get\&Set}{\Q{}[\localTail{}], \itemX{}}==\bot$}\label{piq:swap}
            \State \recoverycode{\pwb{$\Q{}[\localTail{}]$}}; \recoverycode{\psync{}} \label{piq:eqpwb}
            \State \Return \OK{}
        \EndIf
    \EndWhile
\EndFunction
\medskip
\Function{Dequeue}{}
    \While{\True{}} \label{piq:dqloop}
        \State $\localHead{}\gets \Call{Fetch\&Increment}{\Head{}}$\label{piq:dqfaa}
        \State $\itemX{}\gets \Call{Get\&Set}{\Q{}[\localHead{}], \top}$\label{piq:dqswap}
        \If{$\itemX{}\neq \bot$} \label{piq:botcheck}
            \State \recoverycode{\pwb{$\Q{}[\localHead{}]$}}; \recoverycode{\psync{}} \label{piq:dqpwb-1}
            \State \Return $\itemX{}$
        \EndIf
        \If{$\Tail{}\leq \localHead{}+1$}\label{piq:dqempty}
            \State \recoverycode{\pwb{$\Q{}[\localHead{}]$};} \recoverycode{\psync{}} \label{piq:dqpwb-2}
            \State \Return \EMPTY{}
        \EndIf
    \EndWhile
\EndFunction
\medskip
\recoverycode{
\Function{Recovery}{}
    \State $\cntb{} = 0$
    \While{$\cntb{} < n$}
        \If{$\Q{}[\Tail{}] == \bot$}
            $\cntb{} = \cntb{} + 1$
        \Else 
            \mbox{ $\cntb{} = 0$}
        \EndIf
        \State $\Tail{}\gets \Tail{} + 1$
    \EndWhile
    \State $\Tail{}\gets \Tail{} - n + 1$
    \State $\Head{}\gets \Tail{}$ 
    \While{$\Q{}[\Head{}]\neq \top$ \textbf{and} $\Head{}\geq 0$}
        \State $\Head{}\gets \Head{} - 1$
    \EndWhile
    \State $\Head{}\gets \Head{} + 1$
\EndFunction
}
\end{algorithmic}
\end{multicols}
\end{algorithm}

\PIQ\ (Algorithm~\ref{alg:piq}) works like \IQ\ and 
persists just the last write into $\Q$ of each operation. 
This is optimal, but also it ensures low persistence cost
due to low contention~\cite{AB+22}, as \IQ\ ensures that 
each element of $\Q$ is accessed by at most two threads. 
Note that persisting \Head\ and \Tail, instead, could also work, but it would be more
expensive in terms of performance, since these variables are accessed by all threads
and thus they are highly contended
(see Section~\ref{sec:evaluation}, Figure~\ref{fig:piq-tail-no-tail}).

On the other hand, avoiding writing back \Head\ and \Tail, adds complexity, 
and thus also performance overhead, to the recovery function.
To recover \Tail, \PIQ\ searches the array for the first continuous streak of $n$ unoccupied cells, where 
$n$ is the total number of threads in the system.
Note that when a crash happens, some of the active enqueuers may already have inserted their items into $\Q$ (and 
their write-backs to NVM has taken effect) but others might not.
However, all the unoccupied positions between two occupied
are taken from enqueue operations that are active at the time of the crash. 
Since there are $n$ processes running concurrently, it is easy to see that $l < n$.
Thus, there are at most $n-1$ unoccupied cells between two occupied.
Because of this, when the recovery function in \PIQ\ finds the first streak of $n$ unoccupied cells in $\Q$, 
it sets \Tail\ to the first cell of that streak.

To find \Head, \PIQ\ (Algorithm~\ref{alg:piq}) starts from \Tail\ and traverses towards the beginning of $\Q$ until it first 
sees the value $\top$. \Head\ is set to point to the element on the right of this $\top$ in $\Q$
(or to $0$ if such an element does not exist).
This way, there does not exist an element with the value $\top$ between \Head\ and \Tail. 

We now discuss why \PIQ\ ensures durable linearizability. 
\noindent
We denote by $\var{enq}_j$ the enqueue operation that reads $j$ in \Tail\ (line~\ref{piq:faa}) 
in its last iteration of the while loop (line~\ref{piq:loop}) 
and successfully inserts an item $x_j$ in position $\Q[j]$. 
A dequeue operation {\em matches} $\var{enq}_j$ if it reads $j$ in \Head\ (line~\ref{piq:dqfaa}) 
and 
replaces $x_j$ with $\top$ in $\Q[j]$ (line~\ref{piq:dqswap}) .
We denote by $\var{deq}_j$ the matching dequeue operation of $\var{enq}_j$.
Any dequeue that reads $j' > j$ in \Head\ (line~\ref{piq:dqfaa}) 
and 
replaces a value $\not\in \{ \top, \bot\}$ with $\top$ in $\Q[j']$ (line~\ref{piq:dqswap}) is 
called a {\em following} dequeue to $\var{deq}_j$.
An operation {\em belongs to} an epoch $E_k$, 
if it performs a \FAI\ in $E_k$. 
A dequeue operation is {\em successful} if 
it replaces a value $x \not\in \{ \top, \bot\}$ with $\top$ in $\Q{}$ (line~\ref{piq:dqswap}).
We say that an enqueue operation is {\em persisted} if it has executed the \Call{Get\&Set}{} 
(line~\ref{piq:swap}), the return value is $\bot$, and the write-back of the new value has taken 
effect.
A dequeue operation is {\em persisted}, if it has executed the \Call{Get\&Set}{} 
(line~\ref{piq:dqswap}), read a value $x\neq \bot$, and the write-back of $\top$ has taken effect.
Persisted operations terminate if the system does not crash. 

We linearize an enqueue $\var{enq}$ that belongs to an epoch $E_k$, if either
$\var{enq}$ is persisted in $E_k$, or 
there exists a matching dequeue $\var{deq}$ to $\var{enq}$ which is persisted in $E_k$; 
$enq$ is linearized at the last \FAI\ it executes. 
We linearize a dequeue $\var{deq}$ of $E_k$, if either
it is persisted in $E_k$, or
there exist a matching enqueue to $\var{deq}$ that has been included in the linearization
and a following dequeue to $\var{deq}$ of $E_k$ is persisted.
If a dequeue returns \EMPTY, it is linearized when it reads \Tail\ on line~\ref{piq:dqempty}.
The linearization points of successful 
dequeues are assigned 
in the same order as the linearization points of their matching enqueues, so that the FIFO property holds.
Algorithm~\ref{alg:IQlin} assigns linearization points to \PIQ\ in a formal (algorithmic) way.
  
\begin{algorithm}[t]
\scriptsize
\caption{\PIQ\ linearization procedure} \label{alg:IQlin}
\begin{algorithmic}[1]
\Statex Execution $\alpha$ and set of epochs $\mathcal{E}\coloneqq \{E_1, E_2, \ldots\}$ in $\alpha$
\State \LinQ{} auxiliary infinite array
\State $head(\LinQ{})$ accessor to head index of \LinQ{}
\State $tail(\LinQ{})$ accessor to tail index of \LinQ{}
\Statex

\For{each epoch $k=1,2,\ldots$}
	\For{each event $j=1,2,\ldots$ in epoch $E_k$} \label{IQlin:eventloop}
		\If{$e_j = \langle \localTail{}\gets \Call{Fetch\&Increment}{\Tail{}}\rangle$} \label{IQlin:eqfai}
			\State Let $\var{enq}(x)$ be the enqueue that executes $e_j$
			\If{$e_j$ is the last \FAI{} of $\var{enq}(x)$ and $\var{enq}(x)$ is persisted\\
				\hspace*{1em}\textbf{or} $\exists \text{ matching deque } \in E_k$ for $\var{enq}(x)$, which is persisted} 
					\State $\LinQ{}[\localTail{}]\gets \itemX{}$
					\State $tail(\LinQ{})\gets \localTail{} + 1$
					\State Linearize that enqueue operation at $e_j$ \label{IQlin:eqlin}
			\EndIf
		\ElsIf{$e_j$ is a \Tail{} read by $\langle deq\colon per\rangle$ such that\\
				\hspace*{0.5em} $\Tail{}\leq \Head{}$ at the time of the read}
			\State Linearize that dequeue operation at $e_j$
		\EndIf

		\While{$head(\LinQ{}) < tail(\LinQ{})$}
			\State $\localHead{}\gets \min \{i\colon \LinQ{}[i]\neq \bot\}$
			\State Let $\var{deq}()$ be the dequeue whose \FAI{} returns $\localHead{}$
			\If{$\var{deq}()$ not active in $e_1, \ldots, e_j$ of $E_k$}
				\State \textbf{break}
			\EndIf
			\If{$\var{deq}()$ is persisted \textbf{or} $\exists$ a following persisted dequeue to $\var{deq}()$ in $E_k$}
				\State $\itemX{}\gets \LinQ{}[\localHead{}]$
				\State $\LinQ{}[\localHead{}]\gets \bot$
				\State $head(\LinQ{})\gets \localHead{} + 1$
				\State Linearize that dequeue operation at $e_j$ \label{IQlin:dqlin}
			\EndIf
		\EndWhile
	\EndFor
\EndFor
\end{algorithmic}
\end{algorithm}

Roughly speaking, all enqueue and dequeue operations that have been persisted by the time of the 
crash are linearized.
This is needed to ensure durable linearizability, as such operations may have also returned.  
We linearize all active 
non-persisted enqueue operations
that have matching dequeue operations which are persisted. Since all persisted dequeue operations 
are linearized, linearizing also their matching enqueues (even if they are not persisted) is necessary for correctness. 
We linearize an 
active non-persisted dequeue $\var{deq}$ 
only if there is a matching enqueue operation $\var{enq}$ that has been included in the linearization
and a following dequeue $\var{deq}'$ to $\var{deq}$ that is persisted in $E_k$ (and thus it is linearized). 
This is necessary because in the absence of the linearization of $\var{deq}$, 
the  linearization of $\var{deq}'$ would violate the FIFO property of the queue (given that 
$\var{enq}$ has been linearized). 
Note that if we do not assign a linearization point to an enqueue and its 
matching dequeue, the FIFO property of the rest operations is not violated. 
So, we do not linearize $\var{deq}$, if there is no following dequeue that is linearized in $E_k$.

\subsection{Persistent CRQ} \label{sec:pcrq}
%

\begin{algorithm}[t]
\scriptsize
\caption{\PCRQ, code for thread $p_i$.
Ignore code in blue to get the pseudocode for \CRQ.
Instances of $\Q$, \Head, and \Tail\ refer to $crq.\Q$, $crq.\Head$ and $crq.\Tail$.
} \label{alg:pcrq}
\begin{multicols}{2}
\begin{algorithmic}[1]
\Function{Enqueue}{\crqV, $\itemX$}
    \State $\var{closedFlag}:$ integer, initially 0
    \While{\True{}} \label{pcrq:eqloop}
        \State $( \var{cb}, \localTail) \gets \Call{Fetch\&Increment}{\Tail}$ \label{pcrq:eqfaa}
        \If{$\var{cb} == 1$}    \Comment{closed bit is set}
            \If{\recoverycode{$\var{closedFlag} == 0$}}
                \State \recoverycode{\pwb{\Tail}; \psync{}} \label{pcrq:closed-pwb-1}
                \State \recoverycode{$\var{closedFlag} \gets 1$}
            \EndIf
            \State \Return \CLOSED{} \label{pcrq:eqclosed-1}
        \EndIf
        \State $e\gets \Q[\localTail\bmod R]$ \label{pcrq:eqr1} \Comment{get item}
        \State $v\gets e.\var{val}$ \label{pcrq:eqr2} \Comment{get item's value}
        \State $( s,i) \gets ( e.s,e.\var{idx}) $ \label{pcrq:eqr3} \Comment{get safe bit and
            index}
        \If{$v == \bot$} \label{pcrq:eqbotchk}
            \If{$i\leq \localTail$ \textbf{and} ($s == 1$ \textbf{or} $\Head \leq \localTail$) \textbf{and} \\
                \hspace*{0.2cm} \Call{CAS2}{$\Q[\localTail\bmod R],( s,i,\bot) ,( 1,\localTail{},\itemX{}) $}} \label{pcrq:eqcas}
                \State \recoverycode{\pwb{$\Q[\localTail\bmod R]$};} \recoverycode{\psync{}} \label{pcrq:eqpwb}
                \State \Return \OK{} \label{pcrq:eqretOK}
            \EndIf
        \EndIf
        \State $\localHead{}\gets \Head{}$ \Comment{read \Head{}} \label{pcrq:headread}
        \If{$\localTail{}-\localHead{}\geq R$ \textbf{or} \Call{starving}{}}\label{pcrq:eqclosedchk}
            \State \Call{Test\&Set}{$\Tail{}.\var{cb}$} \label{pcrq:eqts}
            \State \recoverycode{\pwb{\Tail};} \recoverycode{\psync{}} \label{pcrq:closed-pwb-2}
            \State \recoverycode{$\var{closedFlag} \gets 1$}
            \State \Return \CLOSED{} \label{pcrq:eqclosedret}
        \EndIf
    \EndWhile
\EndFunction
\medskip
%
\Function{Dequeue}{\crqV}
    \While{\True{}} \label{pcrq:dqouterloop}
        \State $\localHead \gets \Call{Fetch\&Increment}{\Head}$ \label{pcrq:dqfaa}
        \State $\Head_i \gets \localHead + 1$
        \State $e\gets \Q[\localHead\bmod R]$ \label{pcrq:dqr1} \Comment{read item}
        \While{\True{}}
            \State $v\gets e.v$ \label{pcrq:dqr2} \Comment{read item's value}
            \State $( s,i) \gets ( e.s,e.\var{idx})$ \label{pcrq:dqr3}
            \If{$i>\localHead{}$} \label{pcrq:dqichk}
                \textbf{goto} line \ref{pcrq:fix} \label{pcrq:enddqr2}
            \EndIf
            \If{$v\neq \bot$} \label{pcrq:dqemptychk}
                \If{$i == \localHead{}$} \label{pcrq:dqhchk}
                    \If{\Call{CAS2}{$\Q[\localHead\bmod R]$,$( s,\localHead,v)$,$( s,
                    \localHead{}+R,\bot)$}} \Comment{dequeue transition} \label{pcrq:dqtran}
                        \State \recoverycode{\pwb{$\Head_i$};} \recoverycode{\psync{}} \label{pcrq:dqpwbHead}
                        \State \Return $v$ \label{pcrq:dqret}
                    \EndIf
                \Else
                    \If{\Call{CAS2}{$\Q[\localHead\bmod R],( s,i,v) ,( 0,i,v)$}}
                    \label{pcrq:unsafe} \Comment{unsafe transition}
                        \State \textbf{goto} line \ref{pcrq:fix} \label{pcrq:endunsafe}
                    \EndIf
                \EndIf
            \Else
                \If{\Call{CAS2}{$\Q[\localHead\bmod R],( s,i,\bot) ,( s,\localHead+R,
                \bot)$}} \Comment{empty transition} \label{pcrq:empty-1}
                    \State \textbf{goto} line \ref{pcrq:fix} \label{pcrq:endempty}
                \EndIf
            \EndIf
        \EndWhile
        \State $( \var{cb},\localTail{}) \gets$ \Tail \label{pcrq:fix}
        \If{$\localTail{}\leq \localHead{}+1$} \label{pcrq:empty}
            \State \recoverycode{\pwb{$\Head_i$};} \recoverycode{\psync{}} \label{pcrq:dqpwbHead2}
            \State \Call{FixState}{\crqV}
            \State \Return \EMPTY \label{pcrq:endfix}
        \EndIf
    \EndWhile
\EndFunction
\medskip
%
\Function{FixState}{crq}
\While{True}
    \State $h\gets$ \Call{Fetch\&Add}{crq.head, 0}
    \State $t\gets$ \Call{Fetch\&Add}{crq.tail, 0}
    \If{$\text{crq.tail}\neq t$}
        \State \textbf{continue}
    \EndIf
    \If{$h\leq t$}
        \State \Return
    \EndIf
    \If{$\Call{CAS}{\text{crq.tail}, t, h}$}
        \State \Return
    \EndIf
\EndWhile
\EndFunction
\medskip
%
\recoverycode{
\Function{Recovery}{}
\State $R\colon$ size of circular array
\State $\Head \gets \max_{i=1\ldots n} \Head_i$
\State $( \var{cb}, \localTail) \gets \Tail$ \label{pcrq:rec-tail-read}
\State $\Tail\gets ( \var{cb}, 0 )$
\For{$i\gets 0;\ i < R;\ i\texttt{++}$} \label{pcrq:rec-tail-loop}
    \If{$\Q[i].\var{val} \neq \bot$
	    \textbf{and} $\Tail < \Q[i].\var{idx} + 1$} \label{pcrq:rec-tail-occupiedchk}
	        \State $\Tail\gets \Q[i].\var{idx} + 1$ \label{pcrq:rec-tail-set2}
    \ElsIf{$\Q[i].\var{val} == \bot$
    \textbf{and} $\Q[i].\var{idx} \geq R$} \label{pcrq:rec-tail-emptychk}
        \If{$\Tail < \Q[i].\var{idx} - R + 1$}
            \State $\Tail\gets \Q[i].\var{idx} - R + 1$ \label{pcrq:rec-tail-set1}
        \EndIf
    \EndIf
\EndFor
\If{$\Head>\Tail$} $\Tail\gets \Head$ \Comment{Empty queue} \label{pcrq:recemptycheck}
\Else \Comment{$\Head\leq \Tail$}
    \State $\var{max} \gets \Head$ \label{pcrq:recinimax}
    \For{$i\gets \Head$;\ $i \bmod R \neq \Tail \bmod R$;\ $i\texttt{++}$} \label{pcrq:recmaxfor}
        \If{$\Q[i \bmod R].\var{val} == \bot$ \textbf{and}\\
        \hspace*{0.5em} $\Q[i \bmod R].\var{idx} - R > \var{max}$}
        \label{pcrq:recgtmax}
            \State $\var{max}\gets \Q[i \bmod R].\var{idx} - R + 1$
            \label{pcrq:recsetmax}
        \EndIf
    \EndFor
    \State $\Head \gets \var{max}$ \label{pcrq:rec-max-found}
    \State $\var{min} \gets \Tail$ \label{pcrq:rec-min-ini}
    \For{$i\gets \Head$;\ $i \bmod R \neq \Tail \bmod R$;\ $i\texttt{++}$}
        \If{$\Q[i \bmod R].\var{val} \neq \bot$
        \textbf{and} $\Q[i \bmod R].\var{idx} < \var{min}$
        \textbf{and} $\Q[i \bmod R].\var{idx} \geq \Head{}$}
        \label{pcrq:recgteHead}
            \State $\var{min}\gets \Q[i \bmod R].\var{idx}$
        \EndIf
    \EndFor
    \If{$\var{min} < \Tail$} $\Head \gets \var{min}$ \label{pcrq:rec-min-found} \EndIf
\EndIf
\For{$i = \Head - 1;\ i \bmod R \neq \Tail \bmod R;\ i\texttt{--}$} \label{pcrq:resetloop}
    \State $\Q[i \bmod R] \gets ( 1, i+R, \bot)$ \label{pcrq:resetnodes}
\EndFor
\For{$i = 0;\ i < R;\ i\gets i + 1$} $\Q[i].s \gets 1$ \label{pcrq:resetsfloop1} \EndFor
\EndFunction
}
\end{algorithmic}
\end{multicols}
\end{algorithm}

\PCRQ\ (Algorithm~\ref{alg:pcrq}) builds upon \CRQ, 
so we focus on how to support persistence on top of \CRQ. 
As in \PIQ, our goal is to execute a single \pwbi\ instruction in each operation.
%
An enqueue operation, $\var{enq}$, either returns \CLOSED\  
(line~\ref{pcrq:eqclosed-1} or~\ref{pcrq:eqclosedret}),
or it returns \OK\ (line~\ref{pcrq:eqretOK}). As shown in Algorithm~\ref{alg:pcrq}, we 
managed to insert a single pair of \pwbi-\psynci\ instructions in each of these cases. 
Apparently, the new value written by $\var{enq}$ has to be written back to NVM before it 
completes, to ensure durable linearizability.
So, the pair of \pwbi-\psynci\ in line~\ref{pcrq:eqpwb} is necessary.
In order to respect the semantics of a tantrum queue, if $\var{enq}$ returns \CLOSED, all other enqueues 
on this instance of \CRQ\ that will be linearized after $\var{enq}$ should also return \CLOSED. Thus, 
\PCRQ\ has to ensure that the 
state of \Tail\ is \CLOSED\ after recovery. This means that we need to persist the \CLOSED\ value
in \Tail, as soon as it is written there.
This justifies the \pwbi-\psynci\ pair of line~\ref{pcrq:closed-pwb-2}.
However, a pair of \pwbi-\psynci\ is also needed in
line~\ref{pcrq:closed-pwb-1}, to avoid the following bad scenario:
Assume that a thread performs the \TAS\ of line~\ref{pcrq:eqts} and becomes slow before executing the 
\pwbi\ instruction of line~\ref{pcrq:closed-pwb-2}.
Afterwards, another thread reads in \Tail\ the value $1$ for the closed bit and returns \CLOSED\ in
line~\ref{pcrq:eqclosed-1} without persisting the closed bit.
Then, if a crash occurs, at recovery time, \Tail\ will not be closed anymore.
So, it may happen that a subsequent enqueue operation returns \OK,
violating the semantics of a tantrum queue. 

To reduce the persistence cost, we use a known technique~\cite{AB+22,FKK22}, which uses a flag 
($\var{closedFlag}$) 
to indicate whether executing the persistence instructions can be avoided.

Consider now a dequeue operation, $\var{deq}$. 
Assume first that $\var{deq}$ has executed the \CAS\ of line~\ref{pcrq:dqtran} (dequeue transition), 
successfully, and thus, at the absence of a crash, it will return in line~\ref{pcrq:dqret}. 
Our first effort was to follow the same approach as in \PIQ, and persist the triplet written by
$\var{deq}$ into $\Q$. 
Scenario~\ref{scn:uncertaintyProblem} explains why persisting just this triplet is not enough.
\begin{scenario} \label{scn:uncertaintyProblem}
Consider a circular array $\Q$ of size $R=5$ and let its state consist of the following
triplets: 
$\{(1, x_0, 0), (1, x_1, 1), (1, x_2, 2), (1, x_8, 8), (1, \bot, 4)\}$ (see Fig.~\ref{fig:crqunc}
for a graphical illustration). Assume that the values of these triplets are stored in NVM.
If a crash occurs, at recovery time, Algorithm~\ref{alg:pcrq} cannot distinguish between the following two cases: 
(a) the enqueue with index $8$ ($\var{enq}_8$) occurs after the enqueue with index $3$
($\var{enq}_3$) and its matching dequeue ($\var{deq}_3$) have
been completed.
(b) $\var{enq}_8$ stores its item in $\Q[3]$ before $\var{enq}_3$ stores
its own item there and no dequeue operations (at all) have been executed.
In both cases we linearize enqueues $\var{enq}_0$, $\var{enq}_1$, $\var{enq}_2$ and $\var{enq}_8$ 
(as they are persisted operations). 
However, in the first case, we have to also linearize $\var{enq}_3$ and dequeues $\var{deq}_0$, $\var{deq}_1$, $\var{deq}_2$, 
and $\var{deq}_3$, to ensure
the \FIFO\ property.
Thus, at recovery time \Head\ should have a value at least $4$.
Note that in the second case, only enqueue operations are executed, hence at recovery time \Head\ should be 
equal to $0$.
Since these two cases are indistinguishable to \PCRQ, \Recovery\ cannot choose a correct value for \Head.
\end{scenario}

From the above, we figured out that the value of \Head\ should be written back to NVM. 
However, we made the following observation: 
as long as the value of \Head\ is written back to NVM, writing the triplet of line~\ref{pcrq:dqtran}
back is unnecessary. 
If $\var{deq}$ reads $h$ in \Head, then
by writing \Head\ back, no dequeue after the crash
gets index $h$ or less, and $\var{deq}$ is considered to be persisted. 

Writing \Head\ back to NVM only in a successful dequeue is not enough.
To argue that the pair of 
\pwbi-\psynci\ of line~\ref{pcrq:dqpwbHead2} is needed, assume it is omitted from the code. 
Then, 
there is an execution which results in a state where before the crash the last linearized operation
is a dequeue which returns \EMPTY, whereas after the crash we can have a dequeue operation that returns
some item without any enqueue operation to be linearized between these two dequeues. 
This violates the \FIFO\ property.
Note that after any \pwbi-\psynci\ instruction pair, an operation will
return in the absence of a crash. This gives us the optimal one \pwbi-\psynci\ instruction pair for 
each operation.

\begin{figure}[htb]
\centering
\begin{subfigure}{0.33\linewidth}
\centering
\includegraphics{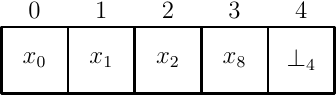}
\caption{State of $\Q$ for Scenario~\ref{scn:uncertaintyProblem}.} \label{fig:crqunc}
\end{subfigure}
\hfil
\begin{subfigure}{0.33\linewidth}
\centering
\includegraphics{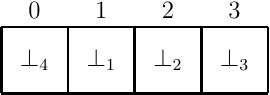}
\caption{State of $\Q$ for Scenario~\ref{scn:pcrq-deq-persists}.}
\label{fig:pcrq-deq-persists}
\end{subfigure}
\hfil\\
\vspace*{2em}
\begin{subfigure}{0.33\linewidth}
\centering
\includegraphics{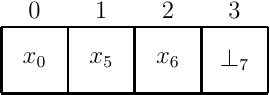}
\caption{State of $\Q$ for Scenario~\ref{scn:pcrq-rec-max}.} \label{fig:movehead}
\end{subfigure}
\caption{Circular array states. On top of the array are the array indexes. Each cell has a value;
the subscript corresponds to the index field value.} \label{fig:scenarios}
\end{figure}

We now describe the recovery function for \PCRQ. 
To find $\Tail$, we traverse the entire array (line~\ref{pcrq:rec-tail-loop}) and search for an 
index that is at least as large
as the maximum index recorded in each occupied cell of $\Q$ 
(lines~\ref{pcrq:rec-tail-occupiedchk}-\ref{pcrq:rec-tail-set2}). 
This is needed to ensure that all items that have been written back
before the crash will be dequeued. 
We have to also consider unoccupied elements of $\Q$, because they can result from 
pairs of enqueues-dequeues which are linearized\footnote{As we explain later, this may occur in two
cases. Either because the system, through a system cache invalidation, writes back the triplet, or because
an obsolete enqueue operation executes a \pwbi\ instruction writing it back.}.
Hence \Tail\ should get a bigger index.
To achieve this, \Tail\ gets the maximum index of unoccupied elements (minus $R$), which is greater 
than its current value (lines~\ref{pcrq:rec-tail-emptychk}-\ref{pcrq:rec-tail-set1}).

We next focus on \Head, 
which, at recovery time, gets a value at least as large as the value of \Head\ that is written back in NVM by the time of the crash. 
This is ensured by lines~\ref{pcrq:recinimax}-\ref{pcrq:rec-min-found}. 
To find the recovered value of \Head, we need to examine the indexes in the unoccupied cells
between current \Head\ and \Tail. We do so because \Recovery\ cannot distinguish if these elements 
are the result of a \CAS\ of line~\ref{pcrq:dqtran} or~\ref{pcrq:empty-1}.
In the former case, by
the way linearization points are assigned, there will be a linearized $\var{deq}$, that we have to
consider in order to ensure durable linearizability. Scenario~\ref{scn:pcrq-deq-persists} provides the details.

\begin{scenario} \label{scn:pcrq-deq-persists}

Let $R=4$ be the size of the circular array $\Q$. Consider an enqueue, $\var{enq}_0(x_0)$, 
which executes its \CAS\ of line~\ref{pcrq:eqcas} successfully and then becomes slow.
Now, a dequeue, $\var{deq}_0()$, becomes active, executes its \CAS\ of 
line~\ref{pcrq:dqtran} successfully removing $x_0$, and then becomes slow.
Next, $\var{enq}_0(x_0)$ finishes its execution 
and returns \OK. Because of the pair of \pwbi-\psynci\ in line~\ref{pcrq:eqpwb},
$\Q[0]$ (which contains the triplet $(s, 4, \bot)$) is written back to NVM.
The state of $\Q$ at this point is illustrated in Fig.~\ref{fig:pcrq-deq-persists}.
Then, a system-wide failure occurs.
At crash time, the value of $\Head$ is equal to $1$, 
but it has not been written back. 

Since $\var{enq}_0(x_0)$ has been completed by the time of the crash, durable linearizability
requires that it is linearized. The question is whether we will linearize $\var{deq}_0()$. 
Since the triplet $(s, 4, \bot)$ has been written back, no dequeue can return $x_0$
after recovery. Thus, if we do not linearize $\var{deq}_0()$, durable linearizability will be violated.
If we linearize $\var{deq}_0()$, the value of \Head\ must be set to $1$ at recovery time. 
By considering only occupied cells of the array this will not be the case. 
Thus, \Recovery\ cannot ignore triplets containing the value $\bot$. 
\qed
\end{scenario}

In \PCRQ, \Recovery\ traverses all elements from $\Head$ to $\Tail$ 
(line~\ref{pcrq:recmaxfor}) 
to find the largest index (minus $R$), $\var{max}$, stored into an unoccupied cell,
that is larger than the current value of $\Head$. We check if the maximum of these values 
is larger than or equal to $\var{max}$, where $\var{max}$ is initialized to the written back value of
\Head\ (line~\ref{pcrq:recinimax}). If it is, we update $\var{max}$ to store the
maximum such value (line~\ref{pcrq:recsetmax}).

\Recovery\ also examines the occupied elements between \Head\ and \Tail. 
We follow a similar approach (as for $\var{max}$) for now finding the minimum index among these elements
(lines~\ref{pcrq:rec-min-ini}-\ref{pcrq:rec-min-found}). 
Scenario~\ref{scn:pcrq-rec-max} explains why doing this is necessary. 
%
%

\begin{scenario} \label{scn:pcrq-rec-max}
Let $R=4$ be the size of the circular array $\Q$. 
Assume that four enqueue operations, $\var{enq}_{0}, \ldots, \var{enq}_{3}$ are executed successfully returning \OK.
Next, $\var{deq}_0$ executes its \FAI\ (line~\ref{pcrq:dqfaa}) and then becomes slow. There are three
more dequeue operations $\var{deq}_{1}, \ldots, \var{deq}_{3}$ that execute successfully. These
deques return items $x_{1}, \ldots, x_{3}$, respectively, and write $\bot_{5}, \ldots, \bot_7$ 
to their respective cells, $\Q[1], \ldots, \Q[ 3]$.
Next, three more enqueue operations $\var{enq}_{4}, \ldots, \var{enq}_{6}$ begin their execution. 
Enqueue $\var{enq}_4$ executes its \FAI\ (line~\ref{pcrq:eqfaa}) and then becomes slow. 
The other two, $\var{enq}_{5}$ and $\var{enq}_{6}$
insert successfully their items, $x_5$ and $x_6$, and return \OK. Now a crash occurs.
Fig.~\ref{fig:movehead} illustrates the state of $\Q$ at crash time.

After the execution of lines~\ref{pcrq:rec-tail-read}-\ref{pcrq:rec-tail-set1} of \Recovery,
\Tail's recovered value is $7$. Moreover, because of $\var{deq}_3$, the value stored in NVM
for \Head\ is $4$ at crash time. After executing lines~\ref{pcrq:recinimax}-\ref{pcrq:recsetmax}, 
\Head's value does not change. Thus, \Head\ points to $\Q[0]$, which contains element $x_0$
whose index $0$ is less than $4$. Since $\var{enq}_{0}$ and $\var{deq}_{1}, \ldots, \var{deq}_{3}$ have terminated,
they have to be linearized. To ensure the FIFO property, $\var{deq}_{0}$ should be linearized as well. 
Thus, the value of \Head\ must be higher than $0$ at recovery time. 
To address this problem, \PCRQ\ sets \Head\ to 
the smallest index of an occupied cell between the current \Head\ and \Tail.
Hence, after lines~\ref{pcrq:recemptycheck}-\ref{pcrq:rec-min-found} 
are executed $\Head = 5$ and $\Tail = 7$. \qed

\remove{
Let the state of the circular array after executing lines~\ref{pcrq:recinimax}-\ref{pcrq:recsetmax} 
of the recovery function of Algorithm~\ref{alg:pcrq} be the one shown in Fig.~\ref{fig:movehead}, 
where $R = 5$. \y{Edw tha prepei na eksigisoume pws proekypse autos o pinakas
kai giati ta Head kai Tail exoun tis times pou exoun, dhladh pws tis phran.
Epishs giati auto to scenario den tha douleye me ton pinaka na periexei ta ekshs:
$x_0, x_6, x_7, \bot_3, \bot_9$?} Note that $\Head = 5$ after executing lines~\ref{pcrq:recinimax}-\ref{pcrq:recsetmax}. 
As explained later we 
will linearize all dequeue operations with indices up until $5$. However, \Head\ points to an element 
with index $0$ which is less than $5$. Since the dequeue with index $0$ will have already been 
linearized by the crash, \Head\ is updated to the smallest index of an occupied cell, \Head\
and \Tail. \qed
}

\remove{Note that $\Tail = 13$ because of a triplet $( s, 12, x_{12})$ at 
$\Q[2]$, and $\Head = 10$ because of a triplet $( s, 9, \bot)$ at $\Q[4]$. As discussed 
previously we have to take into consideration unoccupied nodes for both \Head\ and \Tail.
Let the minimum index of an occupied node, larger than \Head, be $11$. This is possible if at position 
$\Head \bmod R = 0$ there is a triplet $(s, 5, x_5)$, where $x_5$ can be any value (even
$\bot$).}
\end{scenario}

\Recovery\ also initiates all elements of $\Q$ outside the range $\Head \ldots \Tail$ 
with the appropriate value for serving subsequent enqueue operations 
(lines~\ref{pcrq:resetloop}-\ref{pcrq:resetsfloop1}).

\subsubsection{Local Persistence.}
Experiments show (Figure~\ref{fig:plcrq-comp} of Section~\ref{sec:evaluation}) that
having each successful dequeue executing a \pwbi\ instruction on \Head\
is quite costly. The reason for this is that \Head\ is a heavily contended variable. 

\PCRQ\ introduces the following technique to reduce the cost of persisting a variable $\var{var}$
that is highly contended. 
Every thread $p_i$ maintains a local copy,  $\var{var}_i$, of $\var{var}$. 
Thus, there exists an array of $n$ elements, one for each thread, 
where the thread stores its local copy of the variable. This array is maintained in NVM. 
Every time the thread performs an operation on $\var{var}$, 
it also updates its local copy with the value read. 
Then, $p_i$ persists $\var{var}_i$ rather than $\var{var}$. 
Since $\var{var}_i$ is a single-reader, single-writer variable,
persisting it is very cheap~\cite{AB+22}.

For instance,
in \PCRQ, every thread $p_i$ maintains a local copy, $\Head_i$, of \Head,
and persists $\Head_i$ rather than \Head. This happens every time
\Head\ is  to be persisted by \PCRQ
(lines~\ref{pcrq:dqpwbHead} and \ref{pcrq:dqpwbHead2}). 
The recovery function begins by setting the initial value of \Head\ to be the maximum
value found in every element of the array of local copies of the \Head\ values in NVM. 

We believe that this technique, which we call {\em local persistence}, is of general interest, 
and can be applied to reduce the persistence cost of many other algorithms. 

\subsubsection{Durable Linearizability.} We now discuss how to assign linearization points 
to ensure durable linearizability.
We say  that an enqueue (dequeue) operation {\em has index $i$} if it reads $i$ into \Tail\ (\Head)
at its last iteration of the while loop of line~\ref{pcrq:eqloop} (line~\ref{pcrq:dqouterloop}).
We use $\var{enq}_i$ to denote the enqueue with index $i$ that successfully inserts the triplet 
$( 1, i, x_i )$ in cell $\Q[i \bmod R]$.
A dequeue operation {\em matches} $\var{enq}_i$, 
if it reads $i$ in \Head\ (line~\ref{pcrq:dqfaa}); let $\var{deq}_i$ be this dequeue.

An enqueue operation is {\em persisted} in some epoch $E_k$, if it 
has executed successfully the \CAS\ of line~\ref{pcrq:eqcas} and the write-back of the new value has 
taken effect by $c_k$.
A persisted enqueue will terminate successfully if the 
system does not crash. 
Consider now a dequeue operation $\var{deq}$ that has read the value $i$ into $\Head$ the last 
time it read it.
We say that $\var{deq}$ is {\em persisted} in $E_k$, if either some 
value of $\Head \geq i$ has been written back by $c_k$ or a value $\var{idx} \geq i + R$ stored in 
some element of $\Q$ (by a \CAS\ of lines~\ref{pcrq:dqtran},~\ref{pcrq:empty-1}) has been written back by $c_k$.


We assign linearization points to enqueue and dequeue operations of $E_k$
according to the following rules: \\
{\bf Enqueue operations:} Consider an enqueue operation, $\var{enq}$, 
for which the \FAI\ of the last iteration of the while loop
(line~\ref{pcrq:eqloop}) before $c_k$ reads $i$ from \Tail. We linearize $\var{enq}$ in the following 
cases: \\
{\bf (1)} If $\var{enq}$ has been persisted by $c_k$. \\
{\bf (2)} If there exists a dequeue operation, $\var{deq}$, that 
reads $i$ into $\Head$ at its last iteration of the while loop, and $\var{deq}$ has been persisted by $c_k$. \\
{\bf (3)} If $\var{enq}$ has read $1$ in the closed bit of $\Tail$ and this value 
has been written back by $c_k$. \\
{\bf (4)} If $\var{enq}$ has executed successfully the \TAS\ of line~\ref{pcrq:eqts} and the new value 
of $\Tail.\var{cb}$ has been written back by $c_k$. \\
In the first three cases, $\var{enq}$ is linearized at the point that it executed its last \FAI\ 
operation.
In the fourth case, $\var{enq}$ is linearized at the time it executes the \TAS\ of 
line~\ref{pcrq:eqts}.

\noindent
{\bf Dequeue operations:} Consider a dequeue operation, $\var{deq}$, for which the \FAI\ of the last 
iteration of the while loop (line~\ref{pcrq:dqouterloop}) before $c_k$ reads $i$ in \Head. 
We linearize $\var{deq}$ if it has been persisted by $c_k$ and one of the 
following two conditions hold: \\
{\bf (1)} The matching enqueue operation to $\var{deq}$ has successfully executed the \CAS\ of line~\ref{pcrq:eqcas}. \\
{\bf (2)} $\var{deq}$ has read a value $t \leq i+1$ into $\Tail$ (line~\ref{pcrq:fix}). \\
In the first case, $\var{deq}$ is assigned a linearization point either
exactly after the linearization point of the enqueue operation $\var{enq}$ that read $i$ in \Tail,
if $\var{deq}$ is active at that point, or at the first point in which it is active and all dequeues that have smaller 
indices than $\var{deq}$ have been linearized. 
In the second case, $\var{deq}$ is assigned a linearization point at the read of \Tail\ in 
line~\ref{pcrq:fix}.

Algorithm~\ref{alg:PCRQlin} assigns linearization points to \PCRQ\ in a formal (algorithmic) way.
We continue to argue that \PCRQ\ is durably linearizable.  We start with the following lemma.

\begin{algorithm}[t]
\scriptsize
\caption{\PCRQ{} linearization procedure} \label{alg:PCRQlin}
\begin{algorithmic}[1]
\Statex Execution $\alpha$ and set of epochs $\mathcal{E}\coloneqq \{E_1, E_2, \ldots\}$ in $\alpha$
\State \LinQ{}: auxiliary infinite array
\State $\var{head}(\LinQ{})$:  head index of \LinQ{}
\State $\var{tail}(\LinQ{})$: tail index of \LinQ{}
\Statex

\For{each epoch $k=1,2,\ldots$}
    \For{each event $j=1,2,\ldots$ in epoch $E_k$} \label{PCRQlin:eventloop}
        \If{$e_j$ is a \TAS\ of line~\ref{pcrq:eqts} by $\var{enq}(x)$ and its change has been 
        persisted by $c_k$}
            \State Linearize $\var{enq}(x)$ at $e_j$
        \ElsIf{$e_j$ is a \FAI{} of line~\ref{pcrq:eqfaa} by $\var{enq}(x)$ that reads $1$ in the 
        closed bit of $\Tail$ \textbf{and} the closed bit is persisted by $c_k$}
            \State Linearize $\var{enq}(x)$ at $e_j$
        \ElsIf{$e_j$ is a \FAI\ of $\var{enq}(x)$ in $E_k$} \label{PCRQlin:eqfai}
            \If{$e_j$ is the last \FAI\ of $\var{enq}(x)$ and $\var{enq}(x)$ is persisted\\
                \textbf{or} there exists a $\var{deq}$ such that: \\
                a) $\var{deq}$ has executed successfully the \CAS\ of line~\ref{pcrq:dqtran}, \\
                b) $\var{deq}$ reads $i$ into $\Head$ at its last iteration of the while loop \\
                (line~\ref{pcrq:dqouterloop}), and $\var{deq}$ has persisted in $E_k$} 
                \State $\LinQ[\localTail]\gets \itemX$
                \State $tail(\LinQ)\gets \localTail + 1$
                \State Linearize $\var{enq}(x)$ at $e_j$ \label{PCRQlin:eqlin}
            \EndIf
        \ElsIf{$e_j$ is a \Tail{} read of line \ref{pcrq:fix} by $\var{deq}$ that gives an empty queue
            \textbf{and} $\var{deq}$ has been persisted in $E_k$}
            \State Linearize that dequeue operation at $e_j$
        \EndIf

        \While{$head(\LinQ) < tail(\LinQ)$}
            \State $\localHead\gets \min \{i\colon \LinQ[i]\neq \bot\}$
            \State Let $\var{deq}$ be the dequeue whose \FAI\ at its last iteration returns $\localHead$
            \If{$\var{deq}$ not active in $e_1, \ldots, e_j$ of $E_k$}
                \State \textbf{break}
            \EndIf
            \If{$\var{deq}$ has executed successfully the \CAS\ of line~\ref{pcrq:dqtran}, \\
                $\var{deq}$ reads $\localHead$ into $\Head$ at its last iteration of the while loop \\
                (line~\ref{pcrq:dqouterloop}), and $\var{deq}$ has persisted in $E_k$}
                \State $\itemX{}\gets \LinQ{}[\localHead{}]$
                \State $\LinQ{}[\localHead{}]\gets \bot$
                \State $head(\LinQ{})\gets \localHead{} + 1$
                \State Linearize that dequeue operation at $e_j$ \label{PCRQlin:dqlin}
            \EndIf
        \EndWhile
    \EndFor
\EndFor
\end{algorithmic}
\end{algorithm}

\begin{lemma}
\label{lin_deq}
Assume that a dequeue $\var{deq}$ reads a value $i$ into $\Head$ and assume that $\var{deq}$ writes back \Head\ to NVM by a failure $c_k$. 
Then, at recovery time, the following hold: 
(a) \Head\ will have a value greater than $i$, and (b) \Tail\ will have a value greater than $i$.
\end{lemma}

\begin{proof}
Since $\var{deq}$ is persisted, 
either a value of \Head\ greater than $i$ has been written back
by $c_k$, or an index with value at least $i+R$ has been written into $\Q$ 
and into NVM by $c_k$. 

In case (a), 
\Head\ has a value greater than $i$ at the beginning of the execution of the recovery code. 
By inspection of the code, \Head\ is not assigned a value smaller than this initial value. 
In the second case, let $m$ be the index, which is greater than or equal to $i+R$. 
That index has been written by a \CAS\ of line~\ref{pcrq:dqtran} or~\ref{pcrq:empty-1} 
into $\Q$, and it has been written back by $c_k$. We argue that a greater value than $m-R$ will be 
written into \Head\ during the execution of lines~\ref{pcrq:recinimax}-\ref{pcrq:rec-max-found}.
By inspection of the recovery code, $\Tail$ gets a value greater than the maximum index in elements of $Q$
at recovery time. Thus, $\Tail$ has a value greater than $m-R$, because of 
lines~\ref{pcrq:rec-tail-emptychk}-\ref{pcrq:rec-tail-set1}. Due to 
lines~\ref{pcrq:recinimax}-\ref{pcrq:rec-max-found},
$\Head$ will be initialized to a value greater than $m-R$. After recovery, $\Head$ has a value 
greater than $i$.
This proves case (a).

In case (b),
because $\var{deq}$ is persisted, either some value of $\Head > i$ has been written back 
by $c_k$ or a value $\var{idx}\geq i$ has been written in some element of $\Q$ and in NVM by $c_k$.
Assume first that a value of $\Head > i$ has been written back by $c_k$.
It is apparent from lines~\ref{pcrq:recemptycheck}-\ref{pcrq:rec-min-found}, that 
$\Tail \geq \Head$, at recovery time. Thus,
\Tail\ has a value at least $i$ at recovery time.
The other case is also simple. By inspection of Algorithm~\ref{alg:pcrq}, if at some point in time,
an element of $\Q$ stores an index $\var{idx}$, then it will store an index at least as large as 
$\var{idx}$ at all later times. 
At recovery time, \Tail\ gets the maximum among all indices
written in $Q$, and by assumption, a value greater than or equal to $i$
has been written in some element of $Q$ and persisted by $c_k$. \qed
\end{proof}

\noindent
We consider different cases regarding the execution of the algorithm at crash time.

\noindent
{\bf A. Enqueue Operations}

\noindent
We first study enqueues that return closed and then successful enqueues.

\noindent
\textbf{I. Closed enqueue:} Consider an enqueue $\var{enq}(x)$ that executes the 
\TAS\ of line~\ref{pcrq:eqts} setting the value of $\var{Tail.cb}$ to $1$.
If the closed bit is written back by $c_k$, then after the crash any new enqueue operation 
on this instance of $\crqV$ will return \CLOSED, as needed by the semantics of a tantrum queue. 
Suppose the closed bit is not written back by $c_k$. By inspection of the code 
(lines~\ref{pcrq:closed-pwb-1} and~\ref{pcrq:closed-pwb-2}),
this implies that neither $\var{enq}(x)$, nor any other enqueue that has read the closed bit to be $1$
(at its \FAI\ in \Tail\ of line~\ref{pcrq:eqfaa}) have returned. 
Specifically, if one such enqueue operation had returned, then because of the \pwbi-\psynci\ pairs of 
lines~\ref{pcrq:closed-pwb-1} and~\ref{pcrq:closed-pwb-2}, the closed bit would 
have been written back.
According to Algorithm~\ref{alg:PCRQlin} and the linearization rules described above, 
these enqueue operations are not linearized. 
Thus, at recovery time, the closed bit equals $0$ and no enqueue operation that returns \CLOSED\ 
has been linearized in this epoch. 

\noindent
\textbf{II. Successful enqueue:} We now focus on a successful enqueue $\var{enq}_i(x_i)$ with index $i$
and item $x_i$. We proceed by case analysis. In each case, we argue about the following, at recovery time:\\
$\bullet$ If $\var{enq}_i$ and its matching $\var{deq}$ are linearized, then 
          \Head\ and \Tail\ get a value greater than $i$.\\
$\bullet$ If only $\var{enq}_i$ is linearized, then $\Head \leq i < \Tail$, and $x_i$ exists in $Q$, at recovery time. \\
$\bullet$ If $\var{deq}$ that read $i$ in \Head\ is linearized, $\var{enq}_i$ is also linearized.\\
$\bullet$ If neither $\var{enq}_i$, nor its matching $\var{deq}$ are linearized, 
then $x_i$ does not exist in $\Q$ after $c_k$.

\noindent
{\bf (i) Assume first that the triplet $(s, x_i, i)$ stored into $\Q[i \bmod R]$ by $\var{enq}_i(x_i)$ is written back by $c_k$.}
Then, according to Algorithm~\ref{alg:PCRQlin}, $\var{enq}_i(x_i)$ is linearized at its last \FAI.  

We argue that $\Tail$ will have a value greater than $i$ at recovery time. 
By inspection of the recovery code (lines~\ref{pcrq:rec-tail-read}-\ref{pcrq:recemptycheck}), 
\Tail\ gets a value greater than 
the maximum index seen in every occupied element of $\Q$. 
By inspection of Algorithm~\ref{alg:pcrq}, if 
an element of $\Q$ has an index $\var{idx}$, then at all later times it will get an index at 
least that large. 
Since $i$ has been written back by $c_k$, $\Tail$ will get a value
greater than $i$. 
To argue about \Head, we consider the following cases. 

{\bf (a)} If no matching $\var{deq}$ for $\var{enq}_i$ is linearized in $E_k$, then
we argue that $x_i$ resides between $\Head$ and $\Tail$ at recovery time. 
Thus, $x_i$ can be dequeued on a subsequent epoch. 
This is true if $\Head \leq i < \Tail$.
If there is no matching $\var{deq}$, then no dequeue operation has read an index
greater than or equal to $i$, so no such dequeue has been linearized.
Assume now that there exists a matching $\var{deq}$ to $\var{enq}_i$.  
Since $\var{enq}_i$ has been persisted by $c_k$
and $\var{deq}$ is not linearized, by the way we assign linearization points, 
$\var{deq}$ is not persisted. Thus, no value of $\Head$ greater than or equal to $i$, nor an 
$\var{idx}$ (in some element of $\Q$) greater than or equal to $i$, have been written back by $c_k$. 
Thus, $\var{max}$ is initialized to a value less than $i$ (line~\ref{pcrq:recinimax}). 
After the execution of lines~\ref{pcrq:recmaxfor}-\ref{pcrq:rec-max-found}, 
\Head\ continues to have a value smaller than $i$. 
Then, lines~\ref{pcrq:rec-min-ini}-\ref{pcrq:rec-min-found} of recovery are executed. 
Since $\Head < i < \Tail$, it follows that the minimum value will be $i$ 
or smaller. Thus, the value of \Head\ after recovery is less than or equal to $i$.

{\bf (b)} If $\var{enq}_i$'s matching $\var{deq}$ is linearized in $E_k$,
we argue that at recovery, $\Head$ will get a value greater than $i$.
By the way we assign linearization points, $\var{deq}$ is persisted, and 
by Lemma~\ref{lin_deq} it follows that at recovery time, $\Head$ gets a value greater than $i$.

\noindent
{\bf (ii) Assume next that the triplet written in $\Q[i \bmod R]$ is not written back by $c_k$.}
We consider the following cases.

{\bf (a)} If $\var{enq}_i$ is linearized in $E_k$, then 
there exists a $\var{deq}$ which reads $i$ in $\Head$ and has been 
persisted by $c_k$.
By the way linearization points are assigned, $\var{deq}$ is linearized.
By Lemma~\ref{lin_deq}, $\Head$ and \Tail\ are assigned, at recovery time, a value greater than $i$.

{\bf (b)} If $\var{enq}_i$ is not linearized in $E_k$, then
by the way linearization points are assigned, $\var{enq}_i$ cannot be persisted. 
Moreover, no $\var{deq}$ that reads $i$ in $\Head$ by $c_k$ is persisted.
Since $\var{enq}_i$ is not persisted, $x_i$ does not exists in $\Q$ after $c_k$. 
Since no $\var{deq}$ reads $i$ into $\Head$, there is no matching $\var{deq}$ to $\var{enq}_i$
and thus, no such $\var{deq}$ is linearized. Otherwise, $\var{deq}$ exists but has not been persisted. 
By the way we assign linearization points, $\var{deq}$ is not linearized. 

\vspace*{.1cm}
\noindent
{\bf B. Dequeue operations}

We now focus on a dequeue operation, $\var{deq}_i$ with index $i$. Recall that $deq_i$ will succeed to 
dequeue an item from $Q$ only if the index associated with this item is equal to $i$.
We proceed by case analysis. In each case, we argue about the following, in whatever regards the recovery time:\\
$\bullet$ If $\var{deq}_i$ is linearized and would return $x_i \neq \EMPTY$ (if no crash occurs), 
          then its matching enqueue operation, $enq_i(x_i)$, is also linearized. Moreover, 
          the recovered values of \Head\ and \Tail\ are greater than $i$.\\
$\bullet$ If $\var{deq}_i$ is linearized and would return \EMPTY\ (if no crash occurs), 
          then the recovered value of $\Head$ is greater than $i$. 
          Moreover, if an enqueue reads $i$ in $\Tail$, it does not store its item in $Q[i]$.
          It repeats the while loop of line~\ref{pcrq:eqloop} and stores its item in $Q[i']$, where $i' > i$. \\
$\bullet$ If $\var{deq}_i$ is not linearized, whereas $\var{enq}_i$ is linearized, then no dequeue operation with index $i' > i$ is linearized,
	  and the recovered value of $\Head$ is smaller than $i$ whereas the recovered value of \Tail\ is greater than $i$. \\
$\bullet$ If neither $\var{deq}_i$ nor $\var{enq}_i$ are linearized, 
          then $x_i$ does not exist in $\Q$ after $c_k$.
          
\noindent
{\bf A. Dequeue that (would) return \EMPTY:} 
Assume that $\var{deq}_i$ evaluates the condition of the {\tt if} statement of line~\ref{pcrq:empty} to \True.
We consider two cases. 

(i) $\var{deq}_i$ is persisted. Then, by Lemma~\ref{lin_deq}, 
we get that \Head\ 
has a value greater than $i$ at recovery time. 
Since $deq_i$ is persisted, all dequeues whose index is smaller 
than $i$  are also persisted (by definition). 
By the way linearization points are assigned, no operation that 
has index smaller than $i$ will be linearized after the linearization point of $\var{deq}_i$;
moreover, all dequeues with indexes smaller than $i$ have been linearized before $\var{deq}_i$
(as was the case in \crqV). 
Since $deq_i$ reads a value less than $i$ in \Tail\ and it is linearized at the point
that it performed this read, it follows that all enqueue operations with indexes greater 
than $i$ cannot have been linearized before $\var{deq}_i$ (since they perform their 
\FAI\ on \Tail\ at some later point than $\var{deq}_i$'s read of \Tail).
Since \crqV\ is a correct algorithm, it follows that the queue is indeed \EMPTY\
at the linearization point of $\var{deq}_i$. 

Moreover, we argue that if the system crashes before any other enqueue operation is linearized
between the time that $deq_i$ reads \Tail\ and the time that the crash occurs, 
then the value of $\Tail \leq \Head$ at recovery time.
By Lemma~\ref{lin_deq}, \Tail\ is at least $i$ at recovery time.
Assume, by the way of contradiction, that the recovered value of \Tail\ 
is greater than $i$. Since there are no enqueue operations that are linearized between
the time that $\var{deq}_i$ reads \Tail\ and the time that the crash occurs, by the way
linearization points are assigned, it follows that
there does not exist any persisted enqueue with index greater than $i$.
since $\var{deq}_i$ returns \EMPTY, it has not executed a successful dequeue transition.
These imply that there does not exist any occupied cell in $Q$ with index greater than or equal to $i$. 
By inspection of the code (lines~\ref{pcrq:rec-tail-emptychk}-\ref{pcrq:rec-tail-set1}), it follows that 
if the recovered value of \Tail\ is $i' > i$, it is because
there exists an index equal to $i' + R$ in an unoccupied cell.
By inspection of the code (lines~\ref{pcrq:recmaxfor}-\ref{pcrq:recsetmax}), 
\PCRQ\ will examine this cell in order to find the recovered value of \Head.
It follows that $\Head$ will get a value equal to $\Tail$ in this case. 

(ii) $\var{deq}_i$ is not persisted. 
Then, neither a value of $\Head \geq i$ has been written back by $c_k$, 
nor a value $\var{idx} \geq i + R$ written in 
some element of $\Q$ has been written back by $c_k$.
Since $deq_i$ is not persisted, by inspection of the code, 
it follows that $deq_i$ does not cause any change to the state of $Q$. 
By the way linearization points are assigned,
$deq_i$ is not linearized. 
Note also that no enqueue  will be persisted with index $i$. 
(Thus, none of the required claims is violated.)

\noindent
{\bf B. Dequeues that (would) return a value other than \EMPTY.}
%
Assume that $\var{deq_i}$ performs a dequeue transition (by successfully execuiting the \CAS\ of line~\ref{pcrq:dqtran})
replacing the triplet $(s, i, x_i)$  with the triplet $(s,i+R, \bot)$ in  $\Q[i \bmod R]$. 
We proceed by case analysis.

\noindent
{\bf (i) Assume first that $\var{deq_i}$ is persisted by $c_k$.} 
By the way we assign linearization points, both operations, $\var{deq}_i$ and its matching enqueue $\var{enq}_i$ are linearized (in $E_k$).
By Lemma~\ref{lin_deq}, \Head\ and \Tail\ have a value greater than $i$ at recovery time. 

\noindent
{\bf (ii) Assume next that $\var{deq}_i$ is not persisted by $c_k$.}
By definition, neither a value of $\Head \geq i$ has been written back by $c_k$, 
nor a value $\var{idx} \geq i + R$, stored in $\Q$, has been written back by $c_k$.
By the way linearization points are assigned, it follows that neither $\var{deq}$, nor any dequeue operation with 
index $i' > i$, is linearized. 
We consider the following cases.

{\bf (a)} $\var{enq}_i$ is linearized. Since $\var{deq_i}$ is not persisted by $c_k$, by the way linearization points are assigned it follows that 
$\var{enq}_i$ must be persisted by $c_k$. By inspection of the pseudocode, it follows that $x_i$ can only be removed from $Q$ 
by a dequeue transition. Moreover, such a transition can only be executed by $deq_i$. Since $deq_i$ is not persisted,
$x_i$ appears in $\Q$ at recovery time.

We argue that the value of \Tail\ after recovery is greater than $i$. 
By inspection of \Recovery\ (lines~\ref{pcrq:rec-tail-loop}-\ref{pcrq:recemptycheck}), we have that the value of \Tail\ is (at least) the
maximum index (plus one) of the occupied cells of $\Q$. Since $x_i$ appears in $Q$ at recovery time, 
\Tail\ will be assigned a value that is greater than $i$ (lines~\ref{pcrq:rec-tail-read}-\ref{pcrq:rec-tail-set2}). 

We now argue that \Head\ will have a value less than or equal to $i$ at recovery time.
Recall that there is no persisted value of \Head\ greater than $i$, hence $\var{max}$ is initialized to a value less than or equal to $i$ 
(line~\ref{pcrq:recinimax}). Moreover, there is no index value $\var{idx}$ greater than
or equal to $i+R$ written in some (unoccupied) element of $\Q$ (by a 
\CAS\ of line~\ref{pcrq:dqtran} or \ref{pcrq:empty-1}) that has been written back by $c_k$. 
Thus, after the execution of lines~\ref{pcrq:recinimax}-\ref{pcrq:rec-max-found}, 
the value of \Head\ still has a value less than or equal to $i$. 
We now examine lines~\ref{pcrq:rec-min-ini}-\ref{pcrq:rec-min-found}, 
where \Head\ gets the minimum index of the occupied array elements (in range). 
Since $\var{enq}_i$ is persisted, \Head\ is less than or equal to $i$ and \Tail\
is greater than $i$, it follows that the recovered value of \Head\ is less than or equal to $i$. 

{\bf (b)} Assume that $\var{enq}_i$ is not linearized (so, neither $\var{deq_i}$ nor $\var{enq}_i$ are linearized).
By the way linearization points are assigned, $\var{enq}_i$ is not persisted. It follows that $x_i$ does not appear in $\Q$ at recovery time.

{\bf C. Empty and Unsafe transitions.}

\noindent
We now focus on the cases where a dequeue $\var{deq}$ performs an empty or an unsafe transition
(after reading index $h$ in \Head), and does not evaluate the condition 
of the {\tt if} statement of line~\ref{pcrq:empty} to \True. 
Since $\var{deq}$ reads $\localTail$ in \Tail\ (line~\ref{pcrq:fix}) such that $\localTail > h + 1$, 
by inspection of the code, it follows that $deq$
neither changes the value of \Head, nor performs any persistence instruction. 
If no crash occurs, $deq$ will execute another iteration of the while loop of line~\ref{pcrq:dqouterloop},
getting the next available index. Thus, no dequeue with index $h$ will be included in the linearization.
By the way linearization points are assigned, this is so even if a crash occurs by the time that 
$deq$ reads \Head\ again. 
Note that no persisted enqueue will have index $h$.

\textbf{Empty transition:} 
By inspection of the code, 
the enqueue $\var{enq}$ (if it exists) that has read $h$ into \Tail\ 
has not yet executed successfully its \CAS\ of line~\ref{pcrq:eqcas}, 
and $\var{deq}_i$ has executed successfully its \CAS\ of 
line~\ref{pcrq:empty-1}. 
If the triplet $(s, h+R,\bot)$ stored into $\Q[h \bmod R]$, 
when $deq$ executes the \CAS\ of line~\ref{pcrq:empty-1}, is not written back in NVM by $c_k$,
no change on the state of the queue has occured, thus none of the dequeue claims are violated.)

Assume that $(s, h+R,\bot)$ is persisted by $c_k$.
This may cause the linearization of some dequeue operations 
with indexes less than $h$. By the way linearization points are assigned,
there should be matching enqueues of all these dequeues
which will be also linearized. 
By Lemma~\ref{lin_deq}, both $\Head$ and \Tail\ have values greater than the maximum index
of any of these pairs at recovery time. 
It follows that none of the claims mentioned above 
(bullets for dequeues) is violated. 

\textbf{Unsafe transition:}
By inspection of the code, 
the node $\Q[h \bmod R]$ is occupied with an index value $j < h$, and $\var{deq}$
reads $h$ into \Head\ and successfully executes the \CAS\ of line~\ref{pcrq:unsafe}.
Then, according to Algorithm~\ref{alg:pcrq}, we have an
unsafe transition which sets the safe bit $s$ to zero. 
Note that the change of the safe bit is to ensure that the enqueue operation that would read $h$ in
\Tail\ would not enqueue its item in $Q[h \bmod R]$. The unsafe transition does not 
cause any change to the indexes of cells, and it does not cause the linearization of any operaiton. 
Moreover, at recovery time, all safe bits are set to $1$ (line~\ref{pcrq:resetsfloop1}). 
It follows that none of the claims mentioned above 
(bullets for dequeues) is violated.

\subsection{Persistent LCRQ} \label{sec:plcrq}

In this section, we describe the persistent version of \LCRQ{}, called \PLCRQ.
\PLCRQ\ implements a lock-free linked list of nodes, each of which represents 
an instance of \PCRQ. Each \PCRQ\ instance has its own \Head\ and \Tail\ shared variables that encapsulate the
head and the tail of a durably linearizable tantrum queue of finite size, as discussed in Section~\ref{sec:pcrq}. 
\PLCRQ\ eliminates the disadvantage that the \PCRQ\ queue is of finite size 
and transforms it from a tantrum to a FIFO queue.
Specifically, when the currently active instance of \PCRQ\ becomes \CLOSED, a new instance
is appended in the linked list. Moreover, when the queue of a \PCRQ\ instance
becomes \EMPTY, a dequeuer dequeues this instance from the linked list 
and searches for the item to dequeue in the next element of the linked list. 
Two pointers, \First\ and \Last\ point to the first and the last element
of the linked list, which operates as a queue. These pointers are stored in the NVM.

The pseudocode of \PLCRQ\ (Algorithm~\ref{alg:plcrq} in Appendix) closely follows
that of the lock-free queue implementation by Michael and Scott~\cite{MSQUEUE96}. 
An \Call{Enqueue}{$\itemX{}$} operation starts by reading \Last\ (line~\ref{plcrq:eqread1}),
to get a reference to an instance, \crqV, of \PCRQ{}  (of Algorithm~\ref{alg:pcrq}).
If it manages to read the last element of the list (i.e., the next pointer of the node $\crqV$ 
pointed by \Last\ is equal to \NULL), 
the enqueue operation calls 
\Call{Enqueue}{\crqV{}, $\itemX{}$} method (line~\ref{plcrq:eqenq}) to insert $x$ into the queue of \crqV. 
If this method returns \OK{} then
\Call{Enqueue}{$\itemX{}$} is complete and returns also \OK{}.
If \Call{Enqueue}{\crqV{}, $\itemX{}$} returns \CLOSED{} (line~\ref{plcrq:eqenq}), then
\Call{Enqueue}{$\itemX{}$} creates a new node (i.e., a new instance of \PCRQ), 
containing $\itemX{}$ (line~\ref{plcrq:eqnewcrq}),
and tries to append that to the list by executing the \CAS\ of line~\ref{plcrq:eqlastchk}.  
If the \CAS\ succeeds, it tries to also change \Last{} to point to the new \PCRQ{} node
(line~\ref{plcrq:eqmovetail}).
If the \CAS\ fails, another \PCRQ{} node was inserted. Then, \Call{Enqueue}{$\itemX{}$} 
retries by executing one more loop of its while statement (line~\ref{plcrq:eqloop}).

\Call{Enqueue}{$\itemX{}$} does not attempt to append its node if \Last\ does not point
to the last element of the list (line~\ref{plcrq:eqtailchk}). 
It rather tries (line~\ref{plcrq:eqtailupdate}) to update \Last{} to point to the last element,
and starts one more loop (line~\ref{plcrq:eqloop}). 
This scenario might happen if after appending a new node and before updating \Last, a thread becomes slow (or fails).   
Then, \Last\ is {\em falling behind} (i.e., it does not point to the last element of the list), 
and other threads will have to help by setting \Last\ to point
to the last element  of the list, before they try to append their nodes.

A \Call{Dequeue}{}() reads \First{} to get a pointer to the instance, \crqV, of \PCRQ\ 
that is first in the linked list, and calls \Call{Dequeue}{\crqV{}} (lines~\ref{plcrq:dqread}-\ref{plcrq:dqdeq})
in an effort to dequeue an element from that instance. 
If \Call{Dequeue}{\crqV{}} returns some value, then \Call{Dequeue}{} is complete and returns also that value
(lines \ref{plcrq:dqemptychk} and \ref{plcrq:returnv}).
If \Call{Dequeue}{\crqV{}} returns \EMPTY{}, then \Call{Dequeue}{} checks if there is no next \PCRQ{}
node in the list (line~\ref{plcrq:dqnextchk}), and if that is true, then it returns also \EMPTY{}.
If there is another \PCRQ{} node it tries to move the \First{} pointer (line~\ref{plcrq:dqheadmove})
to that node and starts a new loop (line~\ref{plcrq:dqloop}).

\begin{algorithm}[t]
\scriptsize
\caption{Persistent \LCRQ\ (\PLCRQ)} \label{alg:plcrq}
\begin{multicols}{2}
\begin{algorithmic}[1]
\State struct Node \{
\State ~~~~Node *next;
\State ~~~~\PCRQ\ Object \crqV;
\State \};
\vspace{0.05in}\\
\State Node *\textbf{\First{}}, *\textbf{\Last{}} \Comment{both initially point to a Node
with $\var{next}$ being \NULL\ and \crqV\ in initial state}
\\
\Function{Dequeue}{}
	\While{\True{}} \label{plcrq:dqloop}
		\State $\f\gets \First{}$ \label{plcrq:dqread}
		\State $\crqV{}\gets$ \f$\rightarrow$\crqV \label{plcrq:dqread-1}
		\State $v\gets \Call{Dequeue}{\crqV{}}$ \label{plcrq:dqdeq}
		\If{$v\neq \EMPTY{}$} \label{plcrq:dqemptychk}
			\State \Return $v$ \label{plcrq:returnv}
		\EndIf
		\If{\f{}$\rightarrow$\texttt{next}$=\NULL$} \label{plcrq:dqnextchk}
			\State \Return \EMPTY{}
		\EndIf
		\State \Call{CAS}{\First{}, \f{}, \f{}$\rightarrow$\texttt{next}} \label{plcrq:dqheadmove}
	\EndWhile
\EndFunction
\Statex

\Function{Enqueue}{$x$}
	\State Create a new Node \nd\ with $\var{next}$ equal to \NULL, $x$ stored in $\nd.\crqV.\Q{}[0]$,
	$\nd.\crqV.\Tail = 1$, and $\nd.\crqV.\Head = 0$ \label{plcrq:eqnewcrq}\\
	\State \recoverycode{\pwb{$\nd.next, \nd.\crqV.\Q{}[0], \nd.\crqV.\Tail$};} \recoverycode{\psync{};} \label{plcrq:elementpwb}

	\While{\True{}} \label{plcrq:eqloop}
		\State $\ell \gets$ \Last{} \label{plcrq:eqread1}
		\State $\crqV{}\gets \ell$$\rightarrow$\crqV \label{plcrq:eqread2}
		\If{$\ell$$\rightarrow$$\var{next}\neq \NULL{}$} \label{plcrq:eqtailchk}
			\State \recoverycode{\pwb{$\ell$$\rightarrow$next}}; \recoverycode{\psync{}}; \label{plcrq:nextpwb1}
			\State \Call{CAS}{\Last{}, $\ell, \ell$$\rightarrow$next} \label{plcrq:eqtailupdate}
			\State \textbf{continue} \label{plcrq:continue}
		\EndIf
		\If{$\Call{Enqueue}{\crqV{}, \itemX{}}\neq \CLOSED{}$} \label{plcrq:eqenq}
			\State \Return \OK{} \label{plcrq:eqok}
		\EndIf
		\If{\Call{CAS}{$\ell$$\rightarrow$next, \NULL{}, \nd}} \label{plcrq:eqlastchk}
			\State \recoverycode{\pwb{$\ell$$\rightarrow$next};} \recoverycode{\psync{};} \label{plcrq:nextpwb}
			\State \Call{CAS}{$\Last{}, \ell, \nd$} \label{plcrq:eqmovetail}
			\State \Return \OK{} \label{plcrq:ok}
		\EndIf
	\EndWhile
\EndFunction
\Statex

\recoverycode{
\Function{\PLCRQ\ Recovery}{}
\State $\ell \gets$ \First{}
\While{$\ell$ $\neq \Last$}
	\State \Call{Recovery}{$\ell \rightarrow$\crqV} \Comment{From Algorithm~\ref{alg:pcrq}}
	\State $\ell \gets \ell \rightarrow next$
\EndWhile
\While{$\Last$$\rightarrow$next $\neq \NULL$}
	\State \Call{Recovery}{$\Last$$\rightarrow$\crqV} \Comment{From Algorithm~\ref{alg:pcrq}}
	\State \Last $\gets$ \Last$\rightarrow$next
\EndWhile
\State \Call{Recovery}{\Last$\rightarrow$\crqV} \Comment{From Algorithm~\ref{alg:pcrq}}
\EndFunction}
\end{algorithmic}
\end{multicols}
\end{algorithm}

We now discuss the persistence instructions that we added in \LCRQ\ to make it persistent. 
We were able to achieve persistence without adding any persistence instruction to the dequeue code. 
Thus, we focus on enqueue. The \pwbi-\psynci\ pair of line~\ref{plcrq:elementpwb}
is necessary just in case the enqueue returns on line~\ref{plcrq:ok}. However, 
the persistence of \nd.next should be performed 
before the \CAS\ of line~\ref{plcrq:eqlastchk} for the following reason.
Assume that it occurs after the \CAS\ of line~\ref{plcrq:eqlastchk}.
Then, if the system evicts the cache line where the next field of the node on which the \CAS\ 
has been executed resides (i.e., $\ell\rightarrow$next), and before the next field of the appended node, \nd, is persisted,  
then at recovery time, this field may not be \NULL, whereas \nd\ is the last node 
of the recovered list. This jeopardizes the correctness of the algorithm
(for instance the condition of the {\tt if} statement of line~\ref{plcrq:eqtailchk} may succeed,
although \nd\ is the last node in the list).
Note that \nd.\crqV.\Tail\ and \nd.\crqV.$\Q[0]$ can be persisted
after the execution of the \CAS\ of line~\ref{plcrq:eqlastchk}. However, since these are
fields of a struct of type Node, we place them so that they all reside in the same
cache line and can be persisted all together with a single \pwbi.
Therefore, we persist them all together to save on the number of \pwbi s the algorithm performs.       

We now explain why the \pwbi-\psynci\ pair of line~\ref{plcrq:nextpwb} is necessary.
Assume that this pair is omitted.  
Consider an enqueue that returns on line~\ref{plcrq:ok}. If a crash occurs right after the completion of the enqueue,
at recovery time, the node appended by the enqueue and its enqueued item will not be included in the state of the 
linked list. This violates durable linearizability (as the enqueue has been completed, and this it should be linearized). 

We next explain why the \pwbi-\psynci\ pair of line~\ref{plcrq:nextpwb1} is necessary.
Assume that these persistent instructions are omitted. Assume that an enqueue, $\var{enq}_1$
executes up until ine~\ref{plcrq:eqlastchk} and then becomes slow. Now assume that another enqueue,
$\var{enq}_2$ starts its execution, evaluates the condition of the {\tt if} statement 
of line~\ref{plcrq:eqtailchk} to \True, executes lines~\ref{plcrq:eqtailupdate}-\ref{plcrq:continue}, 
and starts a new loop which makes it return on line~\ref{plcrq:eqok}. If a crash occurs at this point, 
the node appended by $\var{enq}_1$, and the items enqueued by both enqueue operations will not appear 
in the recovered state of the linked list. This violates durable linearizability (as $\var{enq}_2$ 
has been completed, and thus it should be linearized).

We now focus on the recovery function. 
Assume that a crash occurs and the system recovers by calling the recovery function once.  
Since \First\ and \Last\ are non-volatile variables, at recovery time, they may have any value they had obtained
by the crash (as the system may have evicted the cache line they reside and written it back to NVM).   

At recovery time, \First\ and \Last\ point either to the first node of the list (if the cache line
in which they reside has never been evicted by the system and written back to NVM) 
or to nodes that has been appended in the linked list by the crash otherwise.  
In the second case, note that the appended nodes may have been dequeued, by the crash, but they are still connected 
to the linked list, i.e., if we start from the initial node in the list and follow next pointers
we will traverse these nodes).

To recover \Last, we start from the value stored for it in NVM (which can be its initial value) 
and we traverse the list following next pointers of its Nodes, until we reach the last node. 
\Last\ will be set to point to this Node. On the way, we call \Recovery\ on the
\crqV\ object of each Node we traverse. \First\ pointer never changes at recovery.
This has a cost in the performance of dequeues that will be executed after the crash,
as they will have to traverse all Nodes until they reach the first such Node whose 
\crqV\ instance contains items to be dequeued.  

\noindent
{\bf Durable Linearizability.}
\PLCRQ\ is durably linearizable. 
We next discuss how to assign linearization points to operations.

{\bf Dequeue operations:} 
Consider a dequeue operation $\var{deq}$. Let \crqV\ be the last instance 
of \PCRQ\ on which $\var{deq}$ calls \Call{Dequeue}{\crqV}. 
$deq$ is linearized in the following cases:
a) \Call{Dequeue}{\crqV} is linearized in \crqV\ and its response is not \EMPTY\ in the linearization of the corresponding execution of \crqV.
b) \Call{Dequeue}{\crqV} is linearized in \crqV, its response is \EMPTY\ in the linearization of the corresponding execution of \crqV,
and $\var{deq}$ returns \EMPTY. 
 
\remove{
if 
If it is at its last loop (line~\ref{plcrq:dqloop}) and 
\Call{Dequeue}{\crqV} (line~\ref{plcrq:dqdeq}) returns either $v\neq \EMPTY$ or
$\crqV==\NULL$ by $c_k$, we assign a linearization point in accordance to the linearization point of the
\Dequeue\ operation of that \crqV\ instance.
}

{\bf Enqueue operations:}  
Consider an enqueue operation $\var{enq}$. Let \crqV\ be the last instance 
of \PCRQ\ on which $\var{enq}$ calls \Call{Enqueue}{\crqV}. 
We linearize an enqueue operation of \PLCRQ\ if the following three conditions hold
\begin{enumerate}
\item If \Call{Enqueue}{\crqV, $x$} returns \EMPTY\ (line~\ref{plcrq:eqenq}), we assign a 
linearization point at the point of the corresponding \Call{Enqueue}{\crqV, $x$} linearization
point.
\item If the change to the
next pointer $\crqV.\var{next}$ (line~\ref{plcrq:eqlastchk}) is written back to NVM and the new 
item that was enqueued is also written back by $c_k$. We assign a linearization point at
the time the \CAS\ operation has taken effect. 
\end{enumerate}

\remove{
Note that the write-back of $\crqV.\var{next}$ can
take effect
either by the \pwbi-\psynci\ pair of line~\ref{plcrq:nextpwb}
or by the system evicting the appropriate cache line. 
Also note that the write-back of the enqueued item can
also take effect
either by the \pwbi-\psynci\ pair of line~\ref{plcrq:elementpwb}
or by the system evicting the appropriate cache line.
}

\section{Performance Evaluation} \label{sec:evaluation}

For the evaluation, we used a 48-core machine (96 logical cores) consisting of 
2 Intel(R) Xeon(R) Gold 5318Y processors with 24 cores each. 
Each core executes two threads concurrently. Our machine is equipped with a 128GB Intel Optane 200 Series Persistent Memory  (DCPMM) and the system is configured in AppDirect mode.
We use the 1.9.2 version of the \textit{Persistent Memory Development Kit}~\cite{PMDK-web}, which
provides the \pwbi\ and \psynci\ persistency instructions.
The operating system is Ubuntu 20.04.6 LTS
(kernel Linux 6.6.3-custom x86\_64) and we use \textit{gcc} v9.4.0. 
Threads were bound in all experiments following a scheduling policy
which distributes the running threads evenly across the machine's NUMA nodes~\cite{FK12ppopp,FK14}.

In the experiments, each run simulates $10^7$ atomic operations in total; each 
thread simulates $10^7/n$ operations (where $n$ is the number of threads). 
We follow a kind of standard approach~\cite{FK11spaa,FK12ppopp,FK14,FriedmanQueue18,DFC,SP21}, 
where each thread performs pairs of \Enqueue\ and \Dequeue\
starting from an empty queue; this kind of experiment avoids performing 
unsuccessful (and thus cheap) operations. 
Experiments where each thread executed random operations
(50\% of each type) did not illustrate significantly different performance trends.

\remove{
Titan 1:
description: DIMM Synchronous Non-volatile LRDIMM 3200 MHz (0.3 ns)
             product: NMB1XXD128GPS
             vendor: Intel
             physical id: 1
             serial: 00000A77
             slot: P1-DIMMB1
             size: 126GiB
             width: 64 bits
             clock: 3200MHz (0.3ns)

}

We compare \PLCRQ\ with the state-of-the-art persistent queue implementations,
namely \PBqueue\ and \PWFqueue~\cite{FKK22}. 
The experimental analysis in~\cite{FKK22}   
compares \PBqueue\ and \PWFqueue\ with the specialized persistent queue implementation 
in~\cite{FriedmanQueue18} (\FHMP), and those recently published in~\cite{SP21}, 
as well as the persistent queue implementations  based 
on the following general techniques and transformations: 
a) \CAPSULES~\cite{NormOptQueue19}, 
b) \Romulus~\cite{CFR18} 
c) \OneFile~\cite{RC+19}, d) \CXPUC~\cite{CFP20eurosys} and \CXPTM~\cite{CFP20eurosys}, and e) \RedoOpt~\cite{CFP20eurosys}.
Experiments in~\cite{FKK22} showed that \PBqueue\ is at least $2$x faster than all these algorithms.
For this reason, we compare the performance of \PLCRQ\ only with the performance
of \PBqueue\ and \PWFqueue~\cite{FKK22}.

\noindent
{\bf Results.}
Figure~\ref{fig:plcrq-comp} shows the performance of \PLCRQ\ against \PBQueue\ and \PWFQueue.
We measure throughput, i.e. number of operations executed per time unit (thus, larger is better), 
as the number of threads increases. We see that \PLCRQ\ (purple line) is 2x faster than \PBqueue, which is the best competitor.  

\begin{figure}[t]
    \centering
    \begin{minipage}[c]{0.35\linewidth}
    \includegraphics[width=\linewidth]{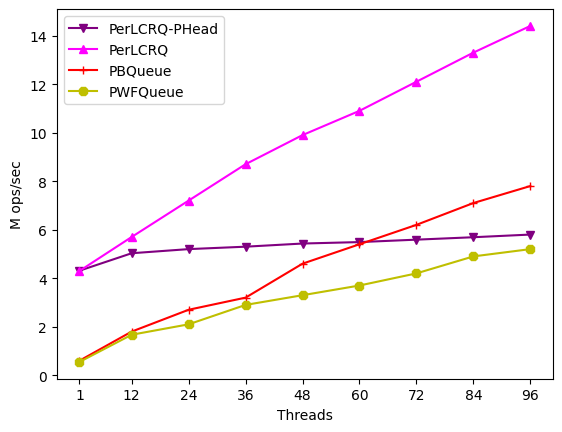}
    \caption{{\scriptsize Performance comparison of \PLCRQ\  with \PBQueue\ and \PWFQueue.}} \label{fig:plcrq-comp}
    \end{minipage}
    \hspace{1cm}
    \begin{minipage}[c]{0.35\linewidth}
    \includegraphics[width=\linewidth]{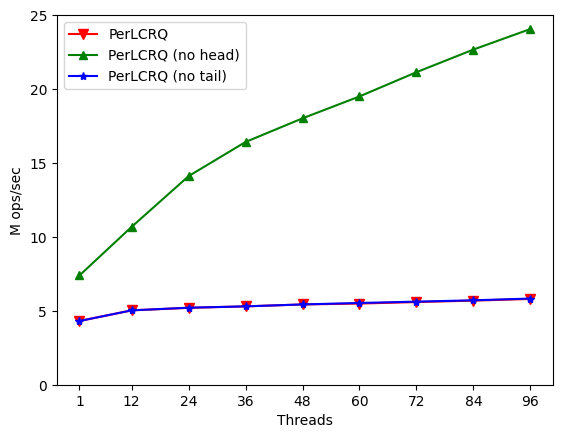}
    \caption{Cost of persisting \Head\ and \Tail\ in \PLCRQ. } \label{fig:over-example}
    \end{minipage}%
\end{figure}

Figure~\ref{fig:plcrq-comp} shows also the cost of persisting different instructions.
Interestingly, although \PLCRQ\ executes only one \pwbi\ per operation, 
persisting shared variable \Head\ is too expensive. As we see,
the version of \PLCRQ\ that persists a shared version of \Head\ (\PLCRQ-PHead)
instead of local copies of \Head\ maintained by each thread, incurs
high performance overhead.  The throughput of \PLCRQ-PHead is much lower than that of \PLCRQ, 
as the number of threads increases, and eventually its performance 
gets surpassed by \PBQueue\ and \PWFQueue. 
Figure~\ref{fig:over-example} provides additional evidence that the drop in performance 
observed in \PLCRQ-PHead, is due to how expensive persisting highly contended shared
variables, such as \Head\ is.  
We denote \PLCRQ\ in which all \pwbi\ instructions on \Head\ have been removed, as \PLCRQ\ {\em (no 
head)} and the \PLCRQ\ in which all \pwbi\ instructions on \Tail\ have been removed, \PLCRQ\ {\em (no 
tail)}. 
From Figure~\ref{fig:over-example}, we see that the persistence cost of persisting \Tail\ is negligible.
This shows that \PLCRQ\ persists \Tail\ rarely (only when it closes the current instance of \CRQ).
Additionally, it provides evidence that the optimization of using the  $\var{closedFlag}$ to avoid 
persisting \Tail\
more than once in the case a \CRQ\ instance needs to close, works well. 
By comparing the throughput of \PLCRQ\ (from Figure~\ref{fig:plcrq-comp}) with that of \PLCRQ\ 
\emph{(no head)}
in Figure~\ref{fig:over-example}, we see the cost of persisting \Head\ even when the algorithm
uses a local copy of \Head\ for each thread. 

\noindent
{\bf Evaluation of the recovery cost.}
We have developed a custom failure framework to evaluate the recovery cost of algorithms.
The framework provides a shared variable called ``recovery\_steps''. All threads "monitor" 
this variable and each operation periodically (or randomly) lowers the value by $1$ step. When it reaches 
$0$, any thread running will cease, effectively simulating a 
crash of all threads. Afterwards, the recovery function is launched 
by some thread; this simulates the system running the recovery function.

The above procedure - a standard run where recovery steps are being decreased, 
a random crash that occurs when these steps become $0$, and the 
recovery function run - is called a {\em cycle}. Each evaluation test has $10$ cycles and 
we measure only the third part of each cycle, which corresponds to the recovery cost. 
The average of the $10$ cycles is the result shown on the graphs in Figures \ref{fig:piq-op} and \ref{fig:piq-size}.

\begin{figure}[htb]
    \centering
    \begin{minipage}[c]{0.32\linewidth}
    \includegraphics[width=\linewidth]{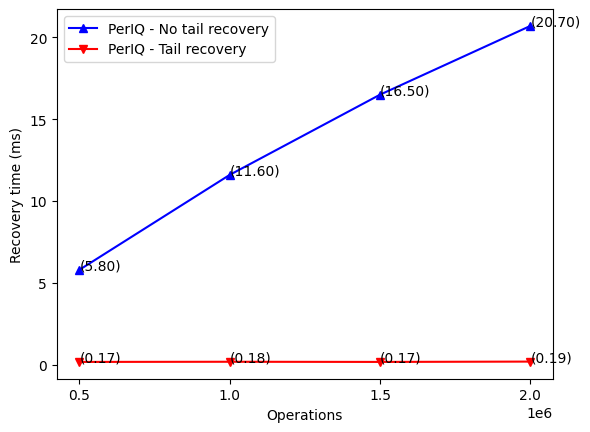}
    \caption{Recovery time of \PIQ\ as the number of operations increases.} \label{fig:piq-op}
    \end{minipage}
    \begin{minipage}[c]{0.32\linewidth}
    \includegraphics[width=\linewidth]{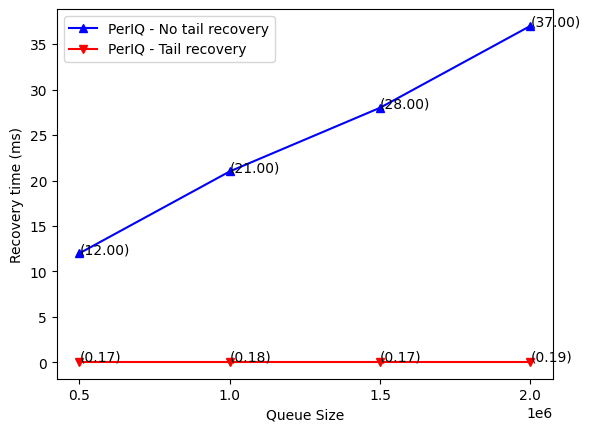}
    \caption{Recovery time of \PIQ\ as the queue size increases.} \label{fig:piq-size}
    \end{minipage}
    \begin{minipage}[c]{0.32\linewidth}
        \includegraphics[width=\linewidth]{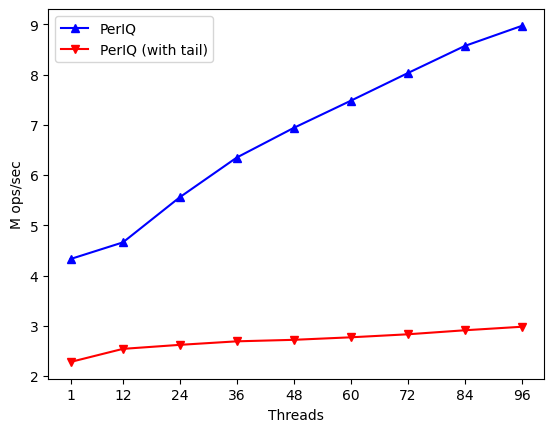}
        \caption{Throughput of \PIQ\ and \PIQ\ (no \Tail). } \label{fig:piq-tail-no-tail}
    \end{minipage}
\end{figure}

\PIQ\ can be designed to illustrate a tradeoff between the persistence cost at
normal execution time (where no failures occur) and the recovery cost. 
The idea is that each thread, periodically, persists \Head\ and \Tail, in addition
to the values it writes in $Q$. This may happen e.g. every time the thread has completed 
the execution of a specific number of operations. 
Algorithm~\ref{alg:enhanced-piq} presents the enhanced pseudocode 
of \PIQ\ to accommodate this variant.

\begin{algorithm}[t]
    \caption{Variant of \Enqueue\ of \PIQ, used to illustrate the tradeoff between the persistence and the recovery cost.}\label{alg:enhanced-piq}
    \begin{algorithmic}
    \State $\nOps_i{}$: counter of thread $p_i$, initially. It counts the the number of enqueue operations $p_i$ has performed. 
    \Statex

    \Function{Enqueue}{$\itemX{}$: Item}
        \While{\True{}} \label{enhanced-piq:loop}
            \State $\localTail{}\gets \Call{Fetch\&Increment}{\Tail{}}$\label{enhanced-piq:faa}

            \If{\recoverycode{$\nOps_i{} == \var{Threshold}$}}
                \State \recoverycode{\pwb{\Tail{}}}
                \State \recoverycode{\psync{}}
                \State \recoverycode{$\nOps_i{} = 0$}
  \LComment{If $\var{Threshold}=\infty$, then persisting \Tail{} is disabled}
            \EndIf
            \If{$\Call{Get\&Set}{\Q{}[\localTail{}], \itemX{}}==\bot$}\label{enhanced-piq:swap}
                \State \recoverycode{\pwb{$\Q{}[\localTail{}]$}} \label{enhanced-piq:eqpwb}
                \State \recoverycode{\psync{}}
                \State \recoverycode{$\nOps{}\gets \nOps{} + 1$}
                \State \Return \OK{}
            \EndIf
        \EndWhile
    \EndFunction
    \Statex

    \end{algorithmic}
\end{algorithm}

Figure~\ref{fig:piq-op} compares the time needed to recover the queue 
after a crash that takes place when a predetermined number of operations have been executed. 
Figure~\ref{fig:piq-size} compares the time  (lower is better) needed to recover the queue based on its size. 
Figure~\ref{fig:piq-tail-no-tail} shows the throughput  (lower is worse) of the \PIQ\ compared to the variation of 
persisting the tail index on each operation. We can see that there is a clear tradeoff between the 
throughput of the persistence of the tail index and how quick the algorithm's recovery function is. 
As shown in Figures \ref{fig:piq-op} to \ref{fig:piq-tail-no-tail}, if we do not persist \Tail, the 
recovery function is slower as the queue gets larger, but the throughput of the \PIQ\ is high. On 
the other hand, if we persist \Tail, the recovery function is very fast, while the throughput of 
\PIQ\ is lower.




\bibliographystyle{splncs04}
\bibliography{refs-copy}

%
\end{document}